\documentclass[aps,pra,superscriptaddress,12pt,longbibliography]{revtex4-2}

\def\cmntsoff{}

\renewcommand{\thesection}{\arabic{section}}

\makeatletter
\renewcommand{\p@subsection}{}
\renewcommand{\p@subsubsection}{}
\makeatother

\usepackage{subfiles}

\usepackage{amsmath,amsfonts,amssymb,amsthm}
\usepackage{mathtools}
\usepackage{bm}
\usepackage{setspace}
\usepackage{optidef}
\usepackage[normalem]{ulem} 
\usepackage{stmaryrd}       
\usepackage{dsfont}         
\usepackage{expdlist}
\usepackage{physics}
\usepackage{braket}
\usepackage{verbatim}
\usepackage{xspace}

\makeatletter
\newcommand{\kbordermatrix}[1]{%
  \vcenter{\hbox{%
    \begingroup
      \m@th
      \normalbaselines
      \everycr{}
      \arraycolsep=6pt
      \tabcolsep=6pt
      \delimiterfactor=900
      \delimitershortfall=5pt
      \def\\{\cr}
      \bordermatrix{#1}
    \endgroup
  }}%
}
\makeatother

\usepackage{pgf}
\usepackage{tikz}
\usetikzlibrary{decorations.pathmorphing}
\usetikzlibrary{arrows,automata}
\usetikzlibrary{shapes.geometric}
\usetikzlibrary{arrows.meta}
\usepackage[utf8]{inputenc}
\usepackage{verbatim}
\usepackage{textcomp} 

\usepackage{graphicx}
\usepackage{xcolor}
\usepackage{hyperref}
\hypersetup{pdfstartview=}
\usepackage{url}

\usepackage{alltt}
\usepackage{moreverb}

\usepackage{xr}

\usepackage{enumitem}
\setlist{nosep} 
\newlist{enuma}{enumerate}{10}
\setlist[enuma]{label*=\arabic*.}
\setlistdepth{10}

\providecommand{\ignore}[1]{}

\newif\ifcmnt
\cmnttrue
\ifdefined\cmntsoff\cmntfalse\fi

\ifcmnt
    \providecommand{\aucmnt}[1]{#1}

    \providecommand{\Pcolor}{\color{gray}}
\else
    \providecommand{\aucmnt}[1]{}

    \providecommand{\Pcolor}{}
\fi

\newcommand{\Ac}[1]{\aucmnt{{\Pcolor \textbf{[}Authors' record: {\color{gray}#1}\textbf{]}}}}

\newcommand{\Prob}{\mathbb{P}}
\newcommand{\Exp}{\mathbb{E}}

\newcommand{\rls}{\mathbb{R}}

\newcommand{\defeq}{\doteq}

\newcommand{\cA}{\mathcal{A}}
\newcommand{\cB}{\mathcal{B}}
\newcommand{\cC}{\mathcal{C}}

\newcommand{\cF}{\mathcal{F}}
\newcommand{\cG}{\mathcal{G}}

\newcommand{\cL}{\mathcal{L}}

\newcommand{\cP}{\mathcal{P}}

\newcommand{\cS}{\mathcal{S}}

\newcommand{\cU}{\mathcal{U}}

\newcommand{\one}{\mathds{1}}

\numberwithin{equation}{section}
\newtheorem{theorem}{Theorem}[section]
\newtheorem*{theorem*}{Theorem}
\newtheorem{lemma}[theorem]{Lemma}
\newtheorem*{lemma*}{Lemma}

\newtheorem{proposition}[theorem]{Proposition}

\newcommand{\opn}{\operatorname}

\let\originalleft\left
\let\originalright\right
\renewcommand{\left}{\mathopen{}\mathclose\bgroup\originalleft}
\renewcommand{\right}{\aftergroup\egroup\originalright}

\newcommand{\ML}{\text{ML}}

\newcommand{\teps}{\tilde{\epsilon}}
\newcommand{\me}{m}

\newcommand{\condfidl}{0.9986}
\newcommand{\condfidh}{0.9993}
\newcommand{\conderrl}{0.0007}
\newcommand{\conderrh}{0.0014}
\newcommand{\totalleakl}{0.0098}
\newcommand{\totalleakh}{0.0108}

\begin{document}
\title{Leakage benchmarking with decohering subgroups in the presence of preparation and measurement errors}
\author{Shawn Geller}
\affiliation{National Institute of Standards and Technology, Boulder, Colorado 80305, USA}
\author{Tongyu Zhao}
\affiliation{National Institute of Standards and Technology, Boulder, Colorado 80305, USA}
\affiliation{Department of Physics, University of Colorado, Boulder, Colorado, 80309, USA}
\author{Alex Kwiatkowski}
\affiliation{National Institute of Standards and Technology, Boulder, Colorado 80305, USA}
\author{Scott Glancy}
\affiliation{National Institute of Standards and Technology, Boulder, Colorado 80305, USA}
\affiliation{Department of Physics, University of Colorado, Boulder, Colorado, 80309, USA}
\author{Raymond W. Simmonds}
\affiliation{National Institute of Standards and Technology, Boulder, Colorado 80305, USA}
\author{Emanuel Knill}
\affiliation{National Institute of Standards and Technology, Boulder, Colorado 80305, USA}
\affiliation{Department of Physics, University of Colorado, Boulder, Colorado, 80309, USA}
\affiliation{Center for Theory of Quantum Matter, University of Colorado, Boulder, Colorado 80309, USA}

\date{\today}
\begin{abstract}
A common problem with physical qubits is that the state can leak outside the computational subspace. To characterize the problem, it is desirable to measure both the leakage probability and the fidelity of gates inside the computational subspace.  Randomized benchmarking is a readily accessible method for estimating average gate fidelities and can be modified to account for leakage. However, in the presence of state preparation and measurement errors (SPAM) there are parameter degeneracies, and fidelities and leakage become underdetermined.  We study the example of qudits with any number of leakage levels under the assumption of gate-independent noise. We show that randomized benchmarking sequences with leakage level decoherence can be analyzed as a Markov process with noisy classical preparations and measurements. This result generalizes to randomized benchmarking with random gates from an arbitrary group that contains a phase-decohering subgroup. We characterize the parameter degeneracy in inferring the Markov process, preparations and measurements. Based on this characterization we develop a measurement and inference method taking advantage of known conservative bounds on SPAM to constrain fidelity, leakage and other parameters of the Markov process, state preparations and measurements.  As part of our study, we experimentally implemented the measurement method for a qubit defined on the one-excitation subspace of two transmons, for which there are two leakage levels. We determined confidence intervals accounting for parameter degeneracy for a number of performance parameters of the system.  For leakage and fidelity, we compare these intervals to those obtained for several conventional methods for leakage benchmarking.
\end{abstract}
\maketitle
\tableofcontents

\section{Introduction}

Quantum information is normally realized as qubits defined as
subsystems of a larger physical system.  The state space of the
physical system therefore contains states orthogonal to the
computational subspace spanned by the states of the qubits. The span
of such orthogonal states is called the leakage subspace of the
physical system. We focus on the case where the the logical states of
the qubits are pure states of the physical system, so that the qubits'
state space is a subspace rather than a more general subsystem.
Storage and manipulation of the qubit states is imperfect and permits
transitions between the computational and leakage subspaces.  A
challenge for characterizing the performance of the qubits is to
measure fidelities of gates, their leakage from the computational to
the leakage subspace, and their seepage from the leakage subspace to
the computational subspace. Errors associated with leakage are often
detectable and therefore easier to control than other errors. As a
result, the fidelity of gates conditional on not having leaked is of
significant interest.

One way to meet the qubit performance characterization challenge in
the presence of leakage is to employ techniques from randomized
benchmarking. Such techniques have been developed and experimentally
implemented.  Theoretical works include
Refs.~\cite{chasseurCompleteRandomizedBenchmarking2015,epstein2014investigating,barnhillExtensionsRandomizedBenchmarking2015,woodQuantificationCharacterizationLeakage2018,andrews2019quantifying,wu2024leakage,Claes2021,ChenRBLeakage2025}
and typically rely on assumptions about the ability to prepare or
measure states of the computational subspace, and they exploit
twirling techniques that decohere the leakage subspace to various
degrees.  These techniques often relate to filtered randomized
benchmarking as developed in generality in
Ref.~\cite{helsen2022generalrb}. Experimental works include studies
with superconducting qubits~\cite{chen2016measuring}, quantum-dot
qubits~\cite{andrews2019quantifying} and atomic
qubits~\cite{muniz2025high,radnaev2025universal,ChenRBLeakage2025,sutherland2025subspace}.
For a helpful overview of the theoretical and experimental work on
leakage benchmarking, see Ref.~\cite{ChenRBLeakage2025}.

A commonly invoked assumption is that it is possible to reliably
prepare and measure states in the computational subspace, with varying
restrictions on the behavior of measurement operators on the leakage
subspace. This assumption constrains state preparation and measurement
error (SPAM) in the physical system. In the absence of such
constraints, with sufficient twirling control as in filtered
randomized benchmarking, it is possible to infer only the eigenvalues
of the twirled noise process. However, as noted in
Ref.~\cite{ChenRBLeakage2025} (see the text after Eq. 49 there), these
eigenvalues can relate poorly to parameters such as average gate
fidelity, leakage and seepage. Specifically, such parameters are not
determined by the probabilities of measurement outcomes. Here we relax
the SPAM constraints so that with reasonable and conservative
assumptions on the quality of preparation and measurement, it is
possible to obtain good confidence intervals on parameters such as
fidelities, leakage and seepage.

An approach for learning leakage and gate fidelities is to apply
randomized benchmarking with Clifford gates on the computational
subspace. If the leaked levels naturally decohere, then according to
the analysis of Ref.~\cite{wu2024leakage} leakage probabilities can be
extracted by analyzing effective classical Markov processes. This
analysis does not support determining average gate fidelities.  To
learn more about the system, we consider the general situation of a
\(d\)-state computational subsystem with \(l\) leaked levels in the
presence of decoherence of the leaked levels. For this situation we
apply sequences of gates randomly chosen from a subspace
representation of a group whose physical realization has
gate-independent noise. The gates need not be unitary, in particular
they may include decoherence. We show that if the group contains a
computational-basis-decohering subgroup, such random sequences can be
modeled as a \(d+l\)-state Markov process with noisy classical state
preparations and measurements. The group need not be a \(2\)-design or
act irreducibly on the computational subspace.  We refer to randomized
benchmarking with such groups and analyzed as a Markov process as
decohering subgroup randomized benchmarking (DSRB).

Most approaches to randomized benchmarking with
leakage~\cite{woodQuantificationCharacterizationLeakage2018,wu2024leakage,hashim2025practical}
require sufficient control to decohere or depolarize the leaked
states. In general randomized benchmarking with groups that are not
\(2\)-designs can be handled with filtered
benchmarking~\cite{helsen2022generalrb}. DSRB enables a full reduction
to a Markov chain model with multiple state preparations and
measurements by taking advantage of a subgroup that phase-decoheres a
basis under twirling with this subgroup. The full benchmarking group's
representation theory implies constraints on the Markov process.  For
example if the group's action is a \(2\)-design, the usual reduction
to one fidelity parameter applies and this fidelity parameter can be
determined in the presence of SPAM. However, for general groups with a
basis-decohering subgroup, the parameters of the Markov process,
preparations and measurements are underdetermined due to the presence
of a ``gauge'' freedom.  DSRB includes mechanisms to account for the
range of consistent values of such underdetermined parameters,
including fidelities, leakage and seepage.  We demonstrate that even
conservative constraints on SPAM can yield useful information about
the Markov process, and thus about the performance of the underlying
gates.  Future work is required to extend existing analysis of
gate-dependent
noise~\cite{proctor2017randomized,wallman2018randomized,helsen2022generalrb}
in randomized benchmarking to DSRB.

\begin{figure}

  \begin{tikzpicture}[line join=round, line cap=round]
    \draw[-{Stealth[length=3.5mm, width=2.5mm]}, line width=1.2pt] (-1.2, -0.8) -- (-1.2, 5.8);
    \node[anchor=east] at (-1.4, 5.2) {Energy};

    \draw[line width=0.8pt] (2.0, 5.2) -- (5.0, 5.2);
    \node[anchor=south] at (3.5, 5.3) {$|ee\rangle$};

    \node[ellipse, draw, thick, line width=1.5pt, minimum width=7.6cm, minimum height=3.0cm] (subspace) at (3.5, 2.5) {};

    \node[anchor=center] at (3.5, 1.9) {Computational subspace};

    \draw[line width=0.8pt] (1.2, 2.5) -- (3.1, 2.5);
    \node[anchor=south] at (2.15, 2.6) {$\ket{0_{L}}=\ket{eg}$};

    \draw[line width=0.8pt] (3.9, 2.7) -- (5.8, 2.7);
    \node[anchor=south] at (4.85, 2.8) {$\ket{1_{L}}=\ket{ge}$};

    \draw[line width=0.8pt] (2.0, -0.2) -- (5.0, -0.2);
    \node[anchor=south] at (3.5, -0.1) {$|gg\rangle$};

    \tikzset{wavytransition/.style={decorate, decoration={coil, aspect=0, amplitude=1.8pt, segment length=5pt, pre length=2pt, post length=2pt}}}

    
    \draw[wavytransition, -{Stealth[length=2.5mm, width=2mm]}] (2.2, 3.91) -- (3.0, 5.2);
    \node[anchor=east] at (2.4, 4.55) {$l_{e}$};

    \draw[wavytransition, -{Stealth[length=2.5mm, width=2mm]}] (4.0, 5.2) -- (4.8, 3.91);
    \node[anchor=west] at (4.6, 4.55) {$s_{e}$};

    \draw[wavytransition, -{Stealth[length=2.5mm, width=2mm]}] (2.2, 1.09) -- (3.0, -0.2);
    \node[anchor=east] at (2.4, 0.45) {$l_{g}$};

    \draw[wavytransition, -{Stealth[length=2.5mm, width=2mm]}] (4.0, -0.2) -- (4.8, 1.09);
    \node[anchor=west] at (4.6, 0.45) {$s_{g}$};

  \end{tikzpicture}\caption{Schematic of the four levels of two
    transmon qubits. The energy spacings are not to scale. The
    levels involving two or more excitations on a transmon are not
    shown.  The dual-rail qubit is defined by the middle two levels,
    which span the computational subspace.  The leakage subspace is
    spanned by the remaining two levels consisting of the ground
    state \(\ket{gg}\) and the excited state \(\ket{ee}\). The
    arrows show the relevant leakage and seepage processes. After
    twirling and taking into account the near absence of direct
    transitions between the two leaked levels, leakage and seepage
    are parameterized by (1) the leakage probabilities \(l_{g}\) and
    \(l_{e}\), which are the probabilities per randomized
    benchmarking step of leaving the computational subspace to
    \(\ket{gg}\) and \(\ket{ee}\), respectively, and (2) the seepage
    probabilities \(s_{g}\) and \(s_{e}\), which are the
    probabilities per step of returning to the computational
    subspace from the states \(\ket{gg}\) and \(\ket{ee}\),
    respectively. }
  \label{fig:fourlevels}

\end{figure}
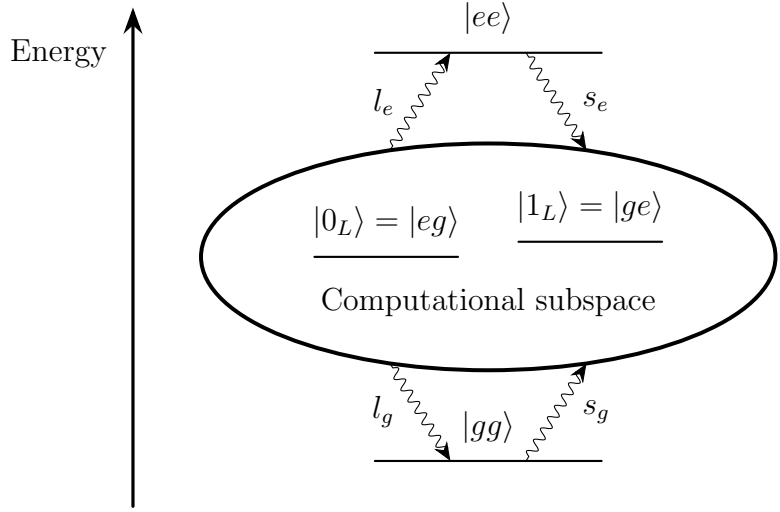

Our work is motivated by the specific situation of one qubit with two
leaked states determined by the ``dual-rail'' qubit
defined on the one-excitation subspace of two transmon qubits.  A
schematic level diagram and the computational subspace are shown in
Fig.~\ref{fig:fourlevels}.  For the dual-rail transmon qubit, \(d=2\)
and \(l=2\). We designed sequences of random one-qubit Clifford gates
for inferring the Markov process and SPAM up to parameter
degeneracies.  Instead of the usual analysis of the probabilities of
the correct outcomes for randomized benchmarking, we directly infer
the Markov process up to degeneracies from the observations. We then
analyze the degeneracies after adding box constraints to the
preparation and measurement matrices to constrain SPAM.  The analysis
yields confidence intervals for the parameters of interest. There are
two types of measurements, and ideally, together they identify which
computational or leakage state was present before the measurement. The
box constraints effectively require that the probability of correctly
identifying one of these states is at least \(0.8\), which is a
conservative bound on the actual quality of the measurement.  We also
assume that the prepared state has fidelity at least \(0.8\) with the
intended computational state. The results from applying the fully
developed analysis to previously obtained experimental data show that
although the constraints are conservative, the confidence intervals
become narrow enough to suggest good average performance of the
one-qubit Clifford gates.  We obtained a fidelity conditional on no
leakage in \([\condfidl, \condfidh]\) (\(68\;\%\) confidence interval)
and a total leakage rate in \([\totalleakl,\totalleakh]\). These
intervals do not shrink to \(0\) in the infinite data limit.

There are several alternative methods currently in use for obtaining
total leakage rates and conditional fidelities. These methods can make
use of only one of the two types of measurements in the experiment.
The Wood-Gambetta
method~\cite{woodQuantificationCharacterizationLeakage2018} does not
take into account measurement errors or leakage levels with
distinguishable behaviors. A number of experiments use variations of a
postselected
method~\cite{Nikolov2023,Lu2023,Chou2024,levineDemonstratingLongCoherenceDualRail2024,muniz2025high,radnaev2025universal,zhao2025improved,Tao2026}
that extend leakage randomized benchmarking as described in
Ref.~\cite{hashim2025practical} Sect. VIII.I.1.  The postselected
method analyzes results as an ordinary randomized benchmarking
experiment after postselecting for no leakage. This method assumes
that there is no seepage. It is also possible to apply one of the
methods described by Chen and Baldwin in
Ref.~\cite{ChenRBLeakage2025}.  None of the alternative methods
account for the gauge freedom or can be used to infer other
parameters. They are also not designed to take advantage of multiple
measurement types.  It is nevertheless informative to compare the
leakage and fidelities obtained to the DSRB results that combine both
types of measurements.  We find that for the dual-rail transmon qubit,
in comparison to DSRB, besides not accounting for SPAM and the
associated gauge freedom, the alternative estimates tend to
underestimate the conditional fidelity and the total leakage. The
overall quality of the fits according to the Wood-Gambetta and the
postselected methods is notably worse than that provided by the DSRB
point estimate. There is only one observable for which the fit quality
of the implemented version of the Chen and Baldwin methods can be
compared to DSRB.  This observable is naturally well behaved in the
Markov model, and this Chen and Baldwin method fits it directly. Not
surprisingly, this fit is better than that of DSRB, but no conclusion
can be drawn about the quality of the inferred conditional fidelity
and total leakage. A version of the postselected method was used to
characterize the dual-rail qubit in an earlier experiment described in
Ref.~\cite{zhao2025improved}, which reports an average conditional
fidelity of \(0.9947(5)\). This is substantially less than the
conditional fidelity obtained by DSRB for the experiment reported
here.  It is also less than what we obtained from applying the
postselected method to the latter experiment.

Besides applications to leakage benchmarking, we expect that DSRB
can be used for many benchmarking scenarios involving quantum states
and processes with physical or enforced decoherence and conservative
constraints on state preparations and measurements.  

\vspace*{\baselineskip}
\noindent\textbf{Outline:}
The paper is structured as follows. We begin by introducing the
motivating experimental system, the dual-rail qubit, in
Sect.~\ref{sec:dual-rail_qubit}. We then dive into the theory of DSRB
in Sect.~\ref{sec:decoheredrb}, where we prove the existence of a
classical Markov model for such randomized benchmarking with
gate-independent errors.  This section also provides a
gauge-independent parametrization of the Markov model, eliminating
redundancy in the operational description with state preparation,
transition matrix, and measurement.  Next, in
Sect.~\ref{sec:inference}, we give an outline of maximum likelihood
inference of the Markov model from data. We transition back to the
motivating experimental system in Sect.~\ref{sec:stochasticmodel},
where we define the physical parameters for the Markov model of the
dual-rail qubit and determine the transition matrix eigenvalues in
terms of these parameters. With this model in hand, in
Sect.~\ref{sec:gaugeopt}, we summarize our method for inferring the
range of physical parameters compatible with the maximum likelihood
inference given the constraints on state preparation and
measurement.  In Sect.~\ref{sec:exresults}, we describe the
experiment that was performed and apply our analysis methods to
obtain leakage and fidelity parameters for the experimental
dual-rail qubit. We also compare the results to those that are
obtained by applying the comparison methods mentioned above and
discuss the quality of the fits to data. We conclude with a
discussion in Sect.~\ref{sec:discussion}. The appendix includes a
description of the experimental design used, the details for the
implementation of maximum likelihood inference, and the
determination of the ranges of parameter values compatible with the
gauge of the ML inferred parameters.

\section{Dual-rail transmon qubit}
\label{sec:dual-rail_qubit}

Our work was motivated by the need for characterizing the dual-rail
transmon qubit described in Ref.~\cite{zhao2025improved} Chpt. 5, and
\cite{Jin2025}. Here we give a high-level overview of the features
relevant for characterization. For details of the device and the
control and measurement techniques used, see the references.  The
device consists of two parametrically tunable transmons with the level
structure shown in Fig.~\ref{fig:fourlevels}. Levels with more than
one excitation on a transmon can be neglected for our purposes. The
dual-rail qubit's state space is spanned by the two single-excitation
levels \(\ket{ge}\) and \(\ket{eg}\).  The transmons are tuned so that
the energy difference between the two levels is insensitive to flux
noise.  The qubit states are defined as \(\ket{0_{L}}=\ket{eg}\) and
\(\ket{1_{L}}=\ket{ge}\).  The leaked states are the ground state
\(\ket{gg}\) and the doubly excited state \(\ket{ee}\).

For characterizing the dual-rail qubit, we apply random sequences
consisting of state preparation, a sequence of gates implemented by
control pulses, and measurement. A single run of the experiment
consisting of a preparation, applying a sequence of gates, followed by
measurement is called a trial.  At the beginning of a trial, the two
transmons have thermally relaxed to a state close to the ground state
\(\ket{gg}\). The qubit is prepared in a state close to \(\ket{1_{L}}\)
by applying a \(\pi\)-pulse to the second transmon, which maps
\(\ket{gg}\) to \(\ket{ge}\). 

For the sequence of gates, parametric drives are used to apply \(\pi\)
and \(\pi/2\) gates rotating between \(\ket{0_{L}}\) and
\(\ket{1_{L}}\).  The phase of the parametric drive determines the
phase of the gate. Each sequence of gates consists of \(n\) steps,
where in each step we apply a random Clifford gate by means of the
fixed pattern of two \(\pi/2\) gates described in
Ref.~\cite{kwiatkowskiThesis2024}, Table 4.2. 

Two different measurements implemented in different trials were
available for each sequence. Both measurements are based on the same
dispersive readout that distinguishes the total number of excitations
in the two-transmon system.  The dispersive readout well distinguishes
the ground state \(\ket{gg}\), the computational subspace spanned by
\(\ket{eg}\) and \(\ket{ge}\) and the excited state \(\ket{ee}\).
Dispersive readout yields two real values that are binned
(Ref.~\citenum{zhao2025improved}, Fig. 5.4 (b)) to obtain one of three
possible outcomes corresponding to the distinguishable subspaces.  The
first measurement directly applies the dispersive readout and resolves
the two leaked states, but not the computational subspace.  The second
measurement first applies a \(\pi\) pulse to the first transmon that
exchanges the qubit states \(\ket{0_{L}}\) with \(\ket{gg}\) and
\(\ket{1_{L}}\) with \(\ket{ee}\), then applies the dispersive
readout. This measurement resolves the two computational states but
not the leaked subspace. The two measurements make it possible to
approximately infer the probabilities of the four transmon states at
the end of the sequence.

Leakage and seepage are due to relaxation and thermal excitation of
the transmons. Since transitions between \(\ket{gg}\) and \(\ket{ee}\)
require two excitations, we assume that the probability of such a
transition during a step of the sequence is negligible. The relative
phases of \(\ket{gg}\) and \(\ket{ee}\) are not tracked in the
experiment and are sensitive to flux noise. We therefore assume that
these phases are randomized in each step, leading to complete
decoherence between \(\ket{gg}\), the computational subspace and
\(\ket{ee}\).  This assumption is unlikely to be perfectly satisfied,
and it is best to actively enforce such decoherence whenever possible.

\section{Decohering subgroup randomized benchmarking}
\label{sec:decoheredrb}

The experiment described in the previous section is an instance of
DSRB.  We begin by considering the general scenario where the physical
state space consists of a \(d\)-dimensional computational subspace
\(\cC\) with computational basis \(\{\ket{k_{L}}\}_{k=1}^{d}\) and
\(l\) orthogonal leaked states \(\{\ket{e_{j}}\}_{j=1}^{l}\).  In the
experiment \(d=2\) and \(l=2\). We denote by \(B\) the basis
consisting of the computational states \(\ket{k_{L}}\) and the leaked
states \(\ket{e_{j}}\).  Let \(\Pi_{c}\) be the projector onto the
computational subspace and \(\Pi_{o;j}\) the projectors onto
\(\ket{e_{j}}\).  For each trial of the experiment, we apply a random
sequence of noisy gates to the prepared state. We assume that the
noise is gate independent so that there exist two noise operations
\(\Lambda_{b}\) and \(\Lambda_{e}\) such that the physical operation
implementing the intended gate's quantum operation \(g\) acts as the
composition of quantum channels
\(\tilde g = \Lambda_{e} g \Lambda_{b}\).  The intended gate preserves
the leaked levels. For the present purposes, we assume that each gate
phase-decoheres the leaked levels and absorb this decoherence into the
intended gate operation.  For the experimental dual-rail qubit, the
decoherence is justified by the lack of phase tracking for the leaked
levels and their sensitivity to flux noise. In general, decoherence
should be actively enforced whenever possible. One way to do so is to
ensure that averaging over the group of intended gates includes such
decoherence. More generally, the theory of DSRB implies that it
suffices for the group of intended gates to include a subgroup that
phase decoheres the members of \(B\).

The assumption on decoherence implies that the intended gates are not
unitary and satisfy
\begin{align}
  g(\rho) &= \Pi_{c}g(\Pi_{c}\rho \Pi_{c})\Pi_{c} +
            \sum_{j=1}^{l} \Pi_{o;j}\rho\Pi_{o;j}.
            \label{eq:intended_gate}
\end{align}
In principle, it is possible to obtain such an action with a starting set of unitary
gates by randomizing the phases of the leaked levels before and after the gate.

For the set of intended gates we choose a subspace representation of a
group \(\cG\) whose members act as quantum channels on operators of
the physical state space.  In particular, the group's identity is
mapped to a projection with respect to the Hilbert-Schmidt inner
product on operators. Twirling with respect to any subgroup of the
physical unitary group yields examples of such projections.  The
intended gate corresponding to \(h\in\cG\) is denoted by the quantum
operation \(g[h]\), where \(g[h]\) satisfies
Eq.~\eqref{eq:intended_gate}, acts unitarily on the computational
subspace, and for all \(h_{1},h_{2}\in\cG\),
\(g[h_{1}h_{2}] = g[h_{1}] g[h_{2}]\). Let \(\cS\) be the subgroup of
\(\cG\) that consists of the members \(h\) of $\cG$ such that \(g[h]\)
fixes each computational basis element. We assume that uniformly
randomly applying \(g[h]\) over \(h\in\cS\) phase decoheres basis
\(B\), and in particular, the computational basis. This is the
defining property of the group used for DSRB. We do not require that
\(g(\cG)\) is a \(2\)- or even a \(1\)-design on the
computational subspace.

Each step of the random sequences except for the last consists of
choosing an independent, uniformly random member \(H\) of \(\cG\) and
applying \(\tilde g[H] =\Lambda_{e}g[H]\Lambda_{b}\). We use
juxtaposition to denote composition of quantum operations, that is, we
do not indicate composition with an explicit binary operator.  We use
capitalization conventions to denote random variables (RVs) such as
\(H\). Instances of random variables are denoted by the corresponding
lower case letters.  For example, an instance or specific value of
\(H\) is denoted by \(h\).

Let \(H_{k}\) be the RV for the random member of \(\cG\) chosen at
step \(k\).  Let \(\bm{H}=(H_{j})_{j=1}^{n}\) be the chosen sequence,
where \(n\) is the sequence length. Let \(T_{k}=H_{k}\ldots H_{1}\) be
the product of the first \(k\) chosen gates.  In the last step,
\(H_{n}\) is chosen uniformly randomly subject to the constraint that
\(T_{n}\) is in a given coset \(h\cS\) of \(\cS\), where we require
that \(g[h]\) permute the computational basis.  Thus \(H_{n}\) depends
on \(T_{n-1}\).  Define \(\bm{H}_{\leq k}=(H_{j})_{j=1}^{k}\) to be the
sequence consisting of the first \(k\) chosen members.  The noisy
implementation of \(\bm{H}\) is defined recursively by
\(\tilde g[\bm{H}_{\leq 0}] =\tilde g[()] = \one\) and
\(\tilde g[\bm{H}_{\leq k+1}] = \tilde g[H_{k+1}]\tilde g[\bm{H}_{\leq
  k}]\).

To analyze the random sequences, we express \(H_{k}\) in terms of a
product of auxiliary RVs \(H'_{k}, S_{k},S'_{k}\) that are not
observed.  Our construction ensures that the RVs
\(H'_{k},S_{k},S'_{k}\) for \(k\in [1,n]\) are independent.  First, we
let \(S_{k}, S'_{k}\) be fresh independent RVs that are uniformly
random on \(\cS\).  The RVs \(H'_{k+1}\) are determined by
\(H_{k+1}\), \(S_{k+1}\), \(S'_{k+1}\) and \(H'_{k}\) according to
\(H_{1}=H'_{1}S_{1}\) and for \(k\geq 1\),
\(H_{k+1}=H'_{k+1}S'_{k+1}S_{k+1}{H'_{k}}^{-1}\).  This implies that
\(H'_{k+1}\) is independent of \(S_{k+1},S'_{k+1}\) and the
\(H'_{j},S_{j},S'_{j}\) for \(j\in [1,k]\). To see this, for
\(1<k+1<n\) write
\(H'_{k+1}= H_{k+1}H'_{k}{S_{k+1}}^{-1}{S'_{k+1}}^{-1}\).  Because
\(H_{k+1}\), \(S_{k+1}\) and \(S'_{k+1}\) are independent and
independent of variables with indices less than \(k+1\), and by
invariance of the uniform distribution (the Haar measure) on \(\cG\),
the probability distribution of \(H'_{k+1}\) is uniform conditional on
\(H'_{k}\), \(S_{k+1}\) and \(S'_{k+1}\). This implies the claimed
independence for \(k+1<n\).  For \(k+1=n\), the condition that
\(T_{n}\in h\cS\) is equivalent to \(H'_{n}\in h\cS\). In general, we
have \(T_{k}\cS=H'_{k}\cS\). This follows by induction: We have
\(T_{1}=H_{1}=H'_{1}S_{1}\), so \(T_{1}\cS=H'_{1}\cS\).  Given that
\(T_{k}\cS=H'_{k}\cS\) and from
\(H_{k+1}=H'_{k+1}S'_{k+1}S_{k+1}{H'_{k}}^{-1}\), it follows that
\(T_{k+1}\cS=H_{k+1}T_{k}\cS=H_{k+1}H'_{k}\cS=H'_{k+1}\cS\). Therefore,
given \(H'_{n-1}\) as well as \(S_{n}, S'_{n}\) and all other RVs with
indices less than \(n\), choosing \(H_{n}\) uniformly so that
\(T_{n}\in h\cS\) is equivalent to choosing \(H_{n}\) uniformly in the
set \(h\cS {H'_{n-1}}^{-1}\). This is in turn equivalent to choosing
\(H'_{n}\) uniformly in the set \(h\cS\), from which it follows that
\(H'_{n}\) also satisfies the claimed independence.

The following lemma generalizes the results for randomized
benchmarking with gate-independent noise to DSRB. We denote the
probability and conditional probability functionals by
\(\Prob(\cdot)\) and \(\Prob(\cdot|\cdot)\), and the expectation and
conditional expectation functionals by \(\Exp(\cdot)\) and
\(\Exp(\cdot|\cdot)\). 

\begin{lemma}\label{lem:twirling}
  Let \(\tilde\Lambda\) be the quantum operation obtained by
  twirling \(\Lambda_{b}\Lambda_{e}\) with \(\cG\) and decoherence
  in the computational basis according to
  \(\tilde\Lambda=\Exp\qty(g[S'H]\Lambda_{b}\Lambda_{e}g[H^{-1}S])\), where
  \(H\), \(S\) and \(S'\) are independent with \(H\) uniformly
  random on \(\cG\), and \(S\) and \(S'\) uniformly random on
  \(\cS\).  Then for \(k\geq 1\) we have
  \begin{align}
    \Exp\qty(\tilde g[\bm{H}_{\leq k}]|H'_{k})
    &= \Lambda_{e}g[H'_{k}]\Exp\qty(g[S'_{k}])\tilde \Lambda^{k-1}\Lambda_{b}.
  \end{align}
\end{lemma}
\begin{proof}
  We have by definition and construction
  \begin{align}
    \tilde g[\bm{H}_{\le 1}] = \Lambda_{e}g[H'_{1}S'_{1}]\Lambda_{b}
    =\Lambda_{e}g[H'_{1}]g[S'_{1}]\tilde\Lambda^{0}\Lambda_{b}
  \end{align}
  which implies the conclusion for \(k=1\).  We proceed by induction.
  Consider \(k\in [1,n-1]\), and assume that
  \(\Exp\qty(\tilde g[\bm{H}_{\leq k}]|H'_{k}) =
  \Lambda_{e}g[H'_{k}]\Exp\qty(g[S'_{k}])\tilde\Lambda^{k-1}\Lambda_{b}\).
  Applying the independence properties of the RVs gives 
  \begin{align}
    \Exp\qty(\tilde g[\bm{H}_{\leq k+1}]\Big| H'_{k+1},H'_{k})\hspace*{-1in}
    &
      \notag\\
    &=\Exp\qty(\Lambda_{e}g[H'_{k+1}S'_{k+1}S_{k+1}(H'_{k})^{-1}]\Lambda_{b}
      \tilde g[\bm H_{\leq k}]\Big|H'_{k+1},H'_{k})
      \notag\\
    &=\Exp\qty(\Lambda_{e}g[H'_{k+1}]g[S'_{k+1}]g[S_{k+1}]g[(H'_{k})^{-1}]
      \Lambda_{b}
      \tilde g[\bm H_{\leq k}]\Big|H'_{k+1},H'_{k})
      \notag\\      
    &=\Lambda_{e}g[H'_{k+1}]\Exp\qty(g[S'_{k+1}])\Exp\qty(g[S_{k+1}])
      g[(H'_{k})^{-1}]
      \Lambda_{b}
      \Exp\qty(\tilde g[\bm H_{\leq k}]\Big|H'_{k})
      \notag\\      
    &=\Lambda_{e}g[H'_{k+1}]\Exp\qty(g[S'_{k+1}])\Exp\qty(g[S_{k+1}])
      g[(H'_{k})^{-1}]\Lambda_{b}\Lambda_{e}
      g[H'_{k}]\Exp\qty(g[S'_{k}])\tilde\Lambda^{k-1}\Lambda_{b}.
  \end{align}
  Taking the expectation over \(H'_{k}\), which is independent of
  \(H'_{k+1}\), gives
  \begin{align}
    \Exp\qty(\tilde g[\bm{H}_{\leq k+1}]\Big| H'_{k+1})\hspace*{-1in}
    &
      \notag\\
    &=\Lambda_{e}g[H'_{k+1}]\Exp\qty(g[S'_{k+1}])\Exp\qty(\Exp\qty(g[S_{k+1}])
      g[(H'_{k})^{-1}]\Lambda_{b}\Lambda_{e}
      g[H'_{k}]\Exp\qty(g[S'_{k}]))\tilde\Lambda^{k-1}\Lambda_{b}
      \notag\\
    &=\Lambda_{e}g[H'_{k+1}]\Exp\qty(g[S'_{k+1}])\Exp\qty(g[S_{k+1}
      (H'_{k})^{-1}]\Lambda_{b}\Lambda_{e}
      g[H'_{k}S'_{k}])\tilde\Lambda^{k-1}\Lambda_{b}
      \notag\\
    &=\Lambda_{e}g[H'_{k+1}]\Exp\qty(g[S'_{k+1}])
      \tilde\Lambda^{k}\Lambda_{b},
  \end{align}
  which implies the conclusion of the lemma for \(k+1\), as required
  to prove the lemma by induction.
\end{proof}

\begin{lemma}
  The twirled noise process \(\tilde\Lambda\) decoheres its input
  and output in the computational and leaked basis \(B\), and
  \(\tilde\Lambda\) acts as a Markov process on \(B\).
\end{lemma}

\begin{proof}
  \(\tilde\Lambda\) includes final pre- and post-processing
  with independent, uniformly random ideally realized members of
  \(\cS\).  Since \(g[h]\) decoheres the leaked levels relative to the
  computational subspace, and by assumption, applying \(g[S]\) for
  uniformly random members of \(\cS\) decoheres the computational
  basis, the claimed decoherence properties follow. This implies that
  \(\tilde\Lambda\) first decoheres its input to mixtures
  of members of \(B\), then transforms members of \(B\) to
  mixtures of members of \(B\) at the output. Therefore \(\tilde\Lambda\)'s
  action is a Markov process on \(B\).
\end{proof}

Let \(\rho\) be a prepared density-matrix state. The average over random sequences
of the state at the end of the 
sequence, conditional on \(H'_{n}\), is given by
\begin{align}
  \Exp\qty(\tilde g[\bm{H}]|H'_{n})(\rho)
  &= \Lambda_{e}g[H'_{n}]\Exp\qty(g[S'_{n}])\tilde \Lambda^{n-1}\Lambda_{b}(\rho).
\end{align}
Because \(\Exp\qty(g[S'_{n}])\) and \(\tilde\Lambda\) decohere their
inputs, only the diagonal elements of \(\Lambda_{b}(\rho)\) in \(B\) matter. As a
result, we can assume that state preparation consists of random
preparation of one of the members of \(B\), where we identify the
prepared state with the state \(\tilde\rho\) obtained by decohering
\(\Lambda_{b}(\rho)\), that is \(\tilde\rho=\Exp(g[S])\Lambda_{b}(\rho)\)
for independent \(S\) uniform on \(\cS\). Since \(\Exp\qty(g[S'_{n}])\)
applied to an already decohered state has no effect,  we have
\begin{align}
  \Exp\qty(\tilde g[\bm{H}]|H'_{n})(\rho)
  &= \Lambda_{e}g[H'_{n}]\tilde \Lambda^{n-1}(\tilde\rho).
\end{align}

To show that the measurement is modeled by a classical noisy
measurement, we use the requirement that \(g[h]\) permute the
computational basis, where \(h\cS\) is the coset of \(T_{n}\).  If
\(P\) is a measurement outcome operator of a positive operator valued
measure, then the outcome probability of \(P\) is given by
\begin{align}
  \Prob(P|\rho)
  &=
    \tr(P\cdot \Lambda_{e}\Exp(g[H'_{n}])\tilde \Lambda^{n-1}(\tilde\rho))
    \notag\\
  &=
    \tr(\Lambda_{e}^{*}(P)\cdot g[h]\Exp(g[S])\tilde \Lambda^{n-1}(\tilde\rho))
    \notag\\
  &=
    \tr(g[h]^{*}(\Lambda_{e}^{*}(P))\cdot
    \tilde\Lambda^{n-1}(\tilde\rho)),
\end{align}
where \(\Sigma^{*}\) denotes the adjoint of quantum operation
\(\Sigma\) and \(S\) is an independent RV that is uniform on \(\cS\).
We used the symbol ``\(\cdot\)'' to explicitly indicate operator
multiplication for clarity.  Since \(g[h]\) permutes the density
matrices associated with members of basis \(B\), for $\ket{i}\in B$ and distinct basis
elements \(\ket{j},\ket{k}\) in \(B\) we have
\(g[h](\ketbra{i})=\ketbra{\sigma(i)}\) for some permutation
\(\sigma\) and
\begin{align}
  \tr(g[h]^{*}(\ketbra{j})\cdot \ketbra{i})
  &= \tr(\ketbra{j}\cdot g[h](\ketbra{i})) =
    \tr(\ketbra{j}\cdot \ketbra{\sigma(i)}) =\delta_{j,\sigma(i)},
    \notag\\
  \tr(g[h]^{*}(\ketbra{j}{k})\cdot \ketbra{i})
  &= \tr(\ketbra{j}{k}\cdot g[h](\ketbra{i})) =
    \tr(\ketbra{j}{k}\cdot \ketbra{\sigma(i)}) = 0.
\end{align}
It follows that \(g[h]^{*}\) permutes the diagonal elements in basis
\(B\) of its argument while creating no overlap with diagonal elements
when applied to the off-diagonal elements. Since
\(\tilde\Lambda^{n-1}(\tilde\rho)\) is diagonal, only diagonal
elements of \(\Lambda_{e}^{*}(P)\) contribute to
\(\Prob(P|\rho)\). The measurement outcome can therefore be
treated as the result of a classical noisy measurement of the output
of the Markov chain applied to a classical, noisy input state, where a
known permutation induced by \(h\) may be applied before the
measurement.

The analysis
above implies that we can model the results from trials of the
experiment by the Markov model consisting of:
\begin{enuma}
\item A \((d+l)\times(d+l)\)-dimensional Markov transition matrix \(Q\).
\item A family of state preparations \(\{\pi_{j}\}_{j=1}^{p}\), where the
  \(\pi_{j}\) are \(d+l\)-dimensional probability column vectors.
\item A family of measurement transition matrices
  \(\{M_{r}\}_{r=1}^{m}\), where the rows of \(M_{r}\) correspond
    to POVM operators \(P\) of the \(r\)'th measurement outcomes. The
    row for POVM operator \(P\) contains the \(d+l\) diagonal entries of
    \(\Lambda_{e}^{*}(P)\).
\end{enuma}
The normalization of POVMs implies that the columns of the \(M_{r}\)
are probability vectors that can be interpreted as the distribution of
outcomes given a particular basis element as the measurement input
state.  We standardize the number of outcomes so that each \(M_{r}\)
has the same number of rows.  An experiment consists of a number of
trials, where each trial involves choosing a state preparation
\(\pi_{j}\), a measurement \(M_{r}\), a computational-basis permuting
\(h\in\cG\), and the number of steps in the random sequences to be
run. The probability distribution from which the observation of the
trial is sampled from is given by the column vector
\begin{align}
  p_{j,r,h,n}&= M_{r}\hat g[h] Q^{n}\pi_{j},
               \label{eq:prob_j_r_h_n}
\end{align}
where \(\hat g[h]\) is the matrix of the permutation of the basis
elements implemented by \(g[h]\).

The problem with SPAM is that if neither the \(M_{r}\) nor the
\(\pi_{j}\) are known, \(Q\) can only be inferred up to a ``gauge''.
Consider any \((d+l)\times(d+l)\)-dimensional invertible matrix \(G\) such that the
\(M'_{r}=M_{r}G^{-1}\) are stochastic, the \(\pi'_{j}=G\pi_{j}\) are
probability distributions and \( \hat g[h]G=G\hat g[h]\) for all
permutations \(\hat g[h]\) used in the experiment. Then we can set
\(Q'=GQG^{-1}\) and satisfy
\begin{align}
  p_{j,r,h,n}&= M'_{r}\hat g[h] Q'^{n}\pi'_{j}.
\end{align}
Any \(G\) that permutes the leaked levels and acts as the identity on the computational
subspace satisfies the given conditions. Not surprisingly, in the absence
of additional assumptions, it is not possible to identify which leaked level is which.

Gauge-free parameters can be obtained as follows.  Let \(\bm{\pi}\) be
the matrix whose columns are the \(\pi_{j}\) and \(\bm{M}\) the matrix
obtained by placing the matrices \(M_{r}\hat g[h]\) one above the
other in a column for the different \(r\) and permutations
\(\hat g[h]\) used.  Then the observable probabilities are given by the
matrices
\begin{align}
  \bm{p}_{n} = \bm{M} Q^{n}\bm{\pi}.
  \label{eq:p_n def}
\end{align}
Assume that the \(\bm{p}_{n}\) are given. Since the dimension of \(Q\)
is \(d+l\), the joint column span \(O\) of the \(\bm{p}_{n}\) is at
most \(d+l\).  We assume that \(O\) is \(d+l\)-dimensional. This
requires that the columns of \(\bm{M}\) are independent, which is the
case if the measurements are complete for inferring density matrices
that are diagonal in \(B\).  Additional conditions need to be
satisfied by the prepared states \(\pi_{j}\) and \(Q\).  Sufficient
such conditions are that the state preparations are spanning, or that
\(Q\) is cyclic on at least one prepared state, meaning that the set
of vectors \(\{Q^{n}\pi_{j}\}_{n\geq 1}\) is spanning for some
\(\pi_{j}\). The column span of the \(\bm{p}_{n}\) can be determined
from the observable probabilities, so we can choose a matrix \(U\)
with \((d+l)\) rows such that for all \(n\),
\(U^{T}U\bm{p}_{n}=\bm{p}_{n}\).  One such matrix is
\((\bm{M}^{T}\bm{M})^{-1/2}\bm{M}^{T}\), but this matrix is not
inferrable. The conditions on \(U\) imply that the column span of
\(U^{T}\) is \(O\). Furthermore, \(UU^{T}\) is the identity. This is
because the identities \(U^{T}U\bm{p}_{n}=\bm{p}_{n}\) imply that for
all \(v\in O\), \(U^{T}U v = v\). In particular \(U^{T}UU^{T}=U^{T}\),
which implies \(UU^{T}=\one\) since \(U^{T}\) has full rank.  These
observations imply that the columns of \(U^{T}\) are an orthonormal
basis of \(O\), so \(U\) is determined up to an orthogonal
transformation on the left.

Let \(Y=(U\bm{M})Q(U\bm{M})^{-1}\) and \(A=(U\bm{M})\bm{\pi}\).
Then
\begin{align}
  U\bm{p}_{n} = Y^{n}A.
  \label{eq:gaugefree}
\end{align}
We claim that \(Y\) and \(A\) are determined by \(U\) and the
\(\bm{p}_{n}\).  Consider a set of indices \(J\) such that the matrix
\(R\) obtained by concatenating the matrices \(U\bm{p}_{n}\) for
\(n\in J\) in a row is full rank. Let \(R'\) be the matrix obtained by
concatenating the \(U\bm{p}_{n+1}\) for \(n\in J\) in a row, with the
same ordering of the concatenated matrices. Then \(R'=YR\), from which
\(Y\) can be determined by applying the right-inverse of \(R\) to both
sides. Once \(Y\) is determined, \(A\) can be determined as
\(Y^{-n}U\bm{p}_{n}\), which does not depend on \(n\geq 1\). 

As noted, \(U\) is determined up to an orthogonal transformation on the
left, in particular \(U=V(\bm{M}^{T}\bm{M})^{-1/2}\bm{M}^{T}\) for some orthogonal
matrix \(V\). Therefore
\begin{align}
  Y = V(\bm{M}^{T}\bm{M})^{1/2}Q(\bm{M}^{T}\bm{M})^{-1/2}V^{T}
\end{align}
for some orthogonal matrix \(V\). This is of the form of a general similarity
transformation \(GQG^{-1}\) of \(Q\). Thus the set of transition
matrices \(Q'\) consistent with the observations is a subset of the
set of conjugates \(G^{-1}YG\) of \(Y\).  Without taking advantage of
constraints on \(\bm{M}\) and \(\bm{\pi}\) and the fact that \(Q\) is a
Markov transition matrix, we can only infer the Jordan decomposition
of \(Q\) from that of \(Y\).

\section{Maximum likelihood inference}
\label{sec:inference}

Here we outline the general strategy for maximum likelihood inference
of model parameters satisfying constraints such as those considered
here.  A detailed description of the strategy used for analyzing the
dual-rail qubit experiment is given in App.~\ref{app:em}.

The analysis of the gauge-free parameters in the previous section
assumes that all measurement probabilities are available at each
sequence length.  In an experiment, only some sequence lengths are
implemented, and not all measurements need be performed at all
implemented sequence lengths. We assume that the implemented
measurements and sequence lengths are sufficient to statistically
infer a transition matrix, preparations and measurements from which
the gauge-free parameters can be determined.  By checking the rank of
the relevant Jacobian as explained in App.~\ref{sec:coptimaldesign},
we verified that for our experimental design, inferrability holds
locally at a reference point in the absence of the additional
constraints on the measurement matrices of Eqs.~\eqref{eq:mr_twirl}
and~\eqref{eq:mr_twirl_1}.

An experiment consists of \(N\) trials.  The \(k\)'th trial of
the experiment consists of preparing the state \(\pi_{s(k)}\), applying a
random sequence of \(n_{k}\) steps and performing the measurement
\(M_{m(k)}\) to obtain the measurement result \(r(k)\).
Let \(\bm{s}=(s(1),\ldots,s(N))\), \(\bm{m}=(m(1),\ldots, m(N))\)
and \(\bm{r}=(r(1),\ldots,r(N))\).

In the first part of the analysis, we use maximum likelihood (ML) to
infer a combination of transition matrix \(Q\), preparations
\(\bm{\pi}\) and measurements \(\bm{M}\) that maximize the probability
of the observed outcomes. We write \(\phi=(\bm{\pi},\bm{M},Q)\) for
the parameters of the model. The ML analysis is designed to respect
additional constraints on \(Q\), \(\bm{\pi}\) and \(\bm{M}\).  Because
of the gauge redundancy, ML solutions \(\phi_{\ML}\) are arbitrary
within a set of gauge equivalent solutions.

A general strategy for finding an ML solution \(\phi_{\ML}\) for
\(\phi\) given the outcomes of the experiment is to use expectation
maximization with techniques developed for hidden Markov models
adapted to the information available in our experiments. See
Ref.~\cite{ephraim2002hidden} for a review of hidden Markov models and
their inference. Expectation maximization for inferring model
parameters can be implemented by an iteration, where in each iteration
an objective function is maximized by means of the Baum algorithm
described in section VI.C of the reference. Hidden Markov models
involve a Markov process applied to an initial state and observed at
each step with a stochastic process. In standard hidden Markov models,
the state at each step is observed with the same, known observation
process, and the goal is to infer the initial state and the transition
matrix.  In our model, there is only one observation at the end of the
sequence, and the observation process is also unknown. The Baum
algorithm can nevertheless be used with a few modifications. The first
modification is to drop the observations at intermediate steps. The
second is to include the parameters for the observation at the last
step in the maximization. The final modification is to include
multiple types of state preparations and observations. The key insight
of the Baum algorithm is that the objective function that is maximized
in each iteration of expectation maximization is an efficiently
computable linear combination of logarithms of affine functions of the
entries of the transition, state preparation and measurement matrices,
where the coefficients of the logarithms are non-negative.  The
transition matrix in Eq.~\eqref{eq:transitionmatrix} below can be
reparameterized so that the entries are linear in the constrained
parameters, with specific linear constraints.  For our application,
the maximization step reduces to a box-constrained simplex problem,
which is an instance of nonlinear resource allocation problems, for
which there are efficient algorithms. See App.~\ref{app:em} for the
reparametrization and the description of the box-constrained simplex
problem and its solution.

In the second part of the analysis we determine for each parameter of
interest an interval that contains the set of the parameter's values
that are achievable within the gauge equivalence class with
conservative constraints on preparation and measurement quality. The
second part of the analysis is outlined in Sect.~\ref{sec:gaugeopt}
after instantiating the Markov model for the dual-rail qubit. It is
detailed in App.~\ref{app:gauge search}.

\section{Markov model for the dual-rail qubit}
\label{sec:stochasticmodel}

For the dual-rail qubit, \(g[\cG]\) restricted to the
computational subspace is the Clifford group. Because this group is a
\(2\)-design, the noise process \(\Lambda_{b}\Lambda_{e}\) is twirled
to a depolarizing channel on the computational subspace. As a result,
the classical Markov process commutes with exchange of the two
computational states. For the classical process, we label the two
computational states as \(0\) and \(1\), representing
\(\ket{0_{L}}\) and \(\ket{1_{L}}\), and the two
leaked levels as \(g\) and \(e\), representing the ground state
\(\ket{gg}\) and the doubly excited state \(\ket{ee}\) in
Fig.~\ref{fig:fourlevels}. Because direct transitions between \(\ket{ee}\)
and \(\ket{gg}\) are negligible, the Markov transition matrix \(Q\) can
be parametrized as
\begin{align}
  Q &=
      \kbordermatrix{
  &0&1&g&e\\
  0&(1-\epsilon ) (1-l_{g} - l_{e}) & \epsilon (1-l_{g} - l_{e}) & \frac{s_{g}}{2} & \frac{s_{e}}{2} \\
  1&\epsilon  (1-l_{g} - l_{e}) & (1-\epsilon ) (1-l_{g} - l_{e}) & \frac{s_{g}}{2} & \frac{s_{e}}{2} \\
  g&l_{g} & l_{g} & 1-s_{g} & 0 \\
  e&l_{e} & l_{e} & 0 & 1-s_{e}
                            }.
                            \label{eq:transitionmatrix}
\end{align}
The conditional fidelity of the dual-rail qubit conditional on not
having leaked is \(1-\epsilon\). The conditional error is \(\epsilon\).
The other parameters are introduced
in Fig.~\ref{fig:fourlevels}. All parameters have values between \(0\)
and \(1\). Because \(l_{g}+l_{e}\) is the total leakage probability,
\(l_{g}+l_{e}\leq 1\).  The traditional quantities of interest are the
conditional error, the unconditional fidelity
\((1-\epsilon)(1-l_{g}-l_{e})\), and the total leakage probability
\(l_{g}+l_{e}\).

In the absence of taking advantage of available constraints, the only
information available from the data is the spectrum of \(Q\).  We
assume that the leakage and seepage probabilities are strictly between
\(0\) and \(1\).  Then the matrix \(Q\) has a complete set of
eigenvalues. Because \(Q\) is symmetric under swap
of the two computational states, \(Q\) preserves the subspace spanned
by \((-1,1,0,0)^{T}\) and the subspace spanned by
\((1/2,1/2,0,0)^{T}, (0,0,1,0)^{T}, (0,0,0,1)^{T}\).  The eigenvalue
of \((-1/2,1/2,0,0)^{T}\) is \(\lambda_{1}=(1-l_{g}-l_{e})(1-2\epsilon)\).
Restricted to the second subspace and in the basis given, \(Q\) acts as
the \(3\times 3\) Markov matrix
\begin{align}
  R &= \begin{pmatrix}
    (1-l_{g}-l_{e})&s_{g}&s_{e}\\
    l_{g}&(1-s_{g})&0\\
    l_{e}&0&(1-s_{e})\\
  \end{pmatrix}.
  \label{eq:Rinvariant}
\end{align}
The steady state of this Markov matrix has eigenvalue \(1\) and is
proportional to \((1,l_{g}/s_{g}, l_{e}/s_{e})^{T}\).  Conjugating
\(R\) by the diagonal matrix \(D\) with diagonal
\((1,\sqrt{s_{g}/l_{g}},\sqrt{s_{e}/l_{e}}) \) makes \(R\) symmetric,
\begin{align}
  DRD^{-1}&=
            \begin{pmatrix}
    (1-l_{g}-l_{e})&\sqrt{l_{g}s_{g}}&\sqrt{l_{e}s_{e}}\\
    \sqrt{l_{g}s_{g}}&(1-s_{g})&0\\
    \sqrt{l_{e}s_{e}}&0&(1-s_{e})\\
  \end{pmatrix}.
\end{align}
which implies that \(R\) is diagonalizable. Let \(\lambda_{-}=1-\delta_{-}\) and \(\lambda_{+}=1-\delta_{+}\)
be the two eigenvalues of \(DRD^{-1}\) different from \(1\). Because
\(\tr(DRD^{-1}-\one)=-(l_{g}+s_{g})-(l_{e}+s_{e})\),
we have \(\delta_{+}+\delta_{-}=l_{g}+l_{e}+s_{g}+ s_{e}\).
The product \(\delta_{+}\delta_{-}\) is given by the coefficient of \(x\)
in \(\det(x\one-(\one-DRD^{-1}))\). This coefficient is given by
\begin{align}
  \delta_{+}\delta_{-}
  &= s_{g}s_{e} %
    + l_{g}s_{e}   
    + l_{e} s_{g}   
    \notag\\
  &= (l_{g}+s_{g})(l_{e}+s_{e})-l_{g}l_{e}
    .
\end{align}
The solutions are
\begin{align}
  \delta_{\pm}&=\frac{1}{2}\qty((l_{g}+l_{e}+s_{g}+ s_{e}) \pm \sqrt{(l_{g}+l_{e}+s_{g}+ s_{e})^{2}
  - 4((l_{g}+s_{g})(l_{e}+s_{e})-l_{g}l_{e})})
                \notag\\
  &=\frac{1}{2}\qty((l_{g}+l_{e}+s_{g}+ s_{e}) \pm \sqrt{((l_{g}+s_{g})-(l_{e}+ s_{e}))^{2}
  +4l_{g}l_{e}})
  .
\end{align}
These eigenvalues constrain the qubit fidelities and the total leakage
probability.

Given the observable probabilities we can determine the eigenvalues of
\(Q\) from those of the gauge free matrix \(Y\) introduced before
Eq.~\eqref{eq:gaugefree}.  It is necessary to determine which
combinations of the parameters \(\epsilon,l_{g},l_{e},s_{g},s_{e}\)
are compatible with the eigenvalues. These parameters can be computed
from the three eigenvalues different from \(1\) and two of the
parameters, \(l_{g}\) and \(s_{g}\).  Let \(t_{g}=l_{g}+s_{g}\) and
\(t_{e}=l_{e}+s_{e}\). Given the eigenvalues we can directly obtain
\(b=\delta_{+}+\delta_{-} = t_{g}+t_{e}\) and
\(c=(\delta_{+}-\delta_{-})^{2}= (t_{g}-t_{e})^{2}+4l_{g}l_{e}\).
From \(\lambda_{1},b,c,l_{g},s_{g}\) we can determine in order
\begin{align}
  t_{g}  &= l_{g}+s_{g},
           \notag\\
  t_{e}  &=b-t_{g},
           \notag\\
  l_{e}  &=\qty(c -(t_{g}-t_{e})^{2})/(4 l_{g}) = \qty(c-(2t_{g}-b)^{2})/(4 l_{g}),
           \notag\\
  s_{e} &= t_{e}-l_{e},
          \notag\\
  \epsilon &= (1-\lambda_{1}/(1-(l_{g}+l_{e})))/2
             .
             \label{eq:lgsg}
\end{align}

We use only one state preparation \(\pi\). It is intended to prepare
the state \(\ket{0_{L}}\), which corresponds to the ideal initial probability
distribution
\begin{align}
  \hat{\pi} &=
              \kbordermatrix{&\\
              0&  1\\
  1& 0\\
  g& 0\\
  e& 0\\
  }
\end{align}
We constrain \(\pi\) by requiring that
\(-0.2\leq \hat{\pi}-\pi\leq 0.2\) entry-wise, so we conservatively constrain the
fidelity of the \(\ket{0_{L}}\) preparation to be at least \(0.8\).

We use two measurements \(M_{1}\) and \(M_{2}\) as described in
Sect.~\ref{sec:dual-rail_qubit}. The \(M_{r}\) are \(3\times 4\)
matrices whose columns are probability vectors. They are intended to implement
the measurement matrices
\begin{align}
  \hat{M}_{1}
  &=
    \kbordermatrix{
  &0&1&g&e\\
  0&0& 0& 1 & 0\\
  1&0&0&0&1\\
  2&1& 1& 0& 0
           },
                            \notag\\
  \hat{M}_{2}
  &=
    \kbordermatrix{
  &0&1&g&e\\
  0&1& 0& 0& 0\\
  1&0& 1& 0 & 0\\
  2&0& 0& 1 & 1
            },
    \label{eq:ideal_measurements}    
\end{align}
The constraints on \(M_{r}\) require
\(-0.2\leq \hat{M}_{r}-M_{r}\leq 0.2\) entry-wise.  These constraints
imply conservative bounds on the probabilities that the first
measurement correctly identifies the leaked states, and the second
measurement correctly identifies the computational states. In addition
they limit the misidentification of the computational subspace for
\(M_{1}\) and of the leakage subspace for \(M_{2}\).

In the experiment, we randomly used one of the two possible
permutations of the two computational states to determine the last
step in the sequences.  We didn't record which final permutation was
applied to the computational states, but in the case of \(M_{2}\) we
applied the inverse of the permutation to the outcomes of the
measurements.  In terms of our model, this twirls \(M_{2}\) by the two
permutations. That is, Eq.~\eqref{eq:prob_j_r_h_n} is replaced by
\begin{align}
  p_{j,2,n}&= \sum_{h}\hat{c_{2}}[h]M_{2}\hat g[h] Q^{n}\pi_{j},
             \label{eq:prob_j_r_h_n'}
\end{align}
where \(\hat{c_{2}}[h]\) is the permutation applied to the outcomes for
\(M_{2}\).  The effect of randomizing the final permutation is that
the \(M'_{2}\) has the form
\begin{align}
  M'_{2} &=
           \kbordermatrix{
  &0&1&g&e\\
  0&a& b& c& d\\
  1&b& a& c & d\\
  2&e& e& f & g
              }.
              \label{eq:mr_twirl}
\end{align}
For \(M_{1}\), forgetting which final permutation was applied has the effect
of standardizing \(M'_{1}\) to the form
\begin{align}
  M'_{1} &=
           \kbordermatrix{
  &0&1&g&e\\
  0&a'& a'& c'& d'\\
  1&b'& b'& e' & f'\\
  2&g'& g'& h' & j'
                 }.
                 \label{eq:mr_twirl_1}
\end{align}

\section{Parameter inference with constraints}
\label{sec:gaugeopt}

Given data for the dual rail qubit, we apply ML inference as outlined
in Sect.~\ref{sec:inference} to obtain an ML fit
\(\phi_{\ML}=(\pi_{\ML},M_{1,\ML},M_{2,\ML},Q_{\ML})\) that satisfies
the assumed constraints. Here we summarize our strategy for inferring
ranges of values of the parameters that are consistent with both
\(\phi_{\ML}\) and the constraints, including the ones due to
randomizing the final permutation that is applied to the computational
states. See App.~\ref{app:gauge search} for details of the strategy.

Consider the set \(F\) of parameter vectors
\(\bm{f}=(f_{1},\ldots, f_{5}) \defeq
(\epsilon,l_{g},l_{e},s_{g},s_{e})\) for which there exist
\(\phi=(\pi,M_{1},M_{2},Q)\) that is gauge equivalent to
\(\phi_{\ML}\) with \(Q\) of the form in
Eq.~\eqref{eq:transitionmatrix} and \(\pi,M_{1},M_{2}\) satisfying the
constraints.  To determine an outer approximation to \(F\), we first
standardize \(\phi\) with a gauge-free representation based on the
formalism of Sect.~\ref{sec:decoheredrb}.  For every pair \(\theta\)
of values of the leakage and seepage parameters \(l_{g},s_{g}\), the
remaining parameters \(\epsilon, l_{e},s_{e}\) can be determined from
the eigenvalues of the gauge-free representation according to
Eq.~\eqref{eq:lgsg}. Write \(\bm{f}(\theta)\) for the parameters so
obtained.  From \(\bm{f}(\theta)\) we obtain \(Q\).  If there exist
\(\pi,M_{1},M_{2}\) that satisfy the constraints and for which
\(\phi\) with \(Q\) obtained from \(\theta\) is gauge equivalent to
\(\phi_{\ML}\), we say that \(\theta\) is valid.  Testing for validity
is difficult.  Instead, we apply algorithmic necessary conditions for
there to exist \(\pi,M_{1},M_{2}\) that satisfy the constraints and
for which \(\phi\) is gauge equivalent to \(\phi_{\ML}\). If the
necessary conditions are satisfied, we call \(\theta\) a feasible
pair. Accordingly, every valid pair is feasible.  Let \(L\) be the set
of such feasible pairs. For each parameter \(f_{i}\), the interval
\([\min(\bm{f}(L)_{i}),\max(\bm{f}(L)_{i})]\) contains the set of
values of \(f_{i}\) consistent with \(\phi_{\ML}\) and the
constraints. Instead of a point estimate of \(f_{i}\), we can report
this interval. To account for statistical uncertainty, we determine
bootstrap confidence intervals that account for the unknown but
constrained gauge. See Sect.~\ref{sec:exresults}.

To implement our strategy, let \(M_{\ML}\) be the \(4\times 4\) matrix
consisting of the first two rows of \(M_{2,\ML}\) followed by the
first two rows of \(M_{1,\ML}\). Given that the constraints are
satisfied in the experiment, \(M_{\ML}\) is with high probability full
rank.  This is because \(M_{1,\ML}\) and \(M_{2,\ML}\) should be close
to \(\hat{M}_{1}\) and \(\hat{M}_{2}\) up to a gauge transformation
on the right. We choose \(U=\one\) in the gauge-free formulation of
Sect.~\ref{sec:decoheredrb} and determine
\(Y=M_{\ML} Q_{\ML}M_{\ML}^{-1}\) and \(A=M_{\ML}\pi_{\ML}\).  Next we
write \(Y=V\Lambda V^{-1}\), where \(\Lambda\) is the full-rank diagonal matrix
with the eigenvalues of \(Y\) on the diagonal, \(V\) is the matrix
whose columns are the right eigenvectors of \(Y\). 

To map out the set \(L\) of feasible pairs \(\theta\), we first
determine the constraints on \(\theta=(l_{g},s_{g})\) induced by the
requirement that the matrix \(Q\) determined by \(\bm{f}(\theta)\) is
a Markov transition matrix. This determines a region
\(\cA^{\mathrm{phys}}\) of physically sensible values of
\(\theta\). We obtain an enclosing band \(\cB\) of this region to
which we confine the search for feasible pairs.  We place a grid of
candidate pairs in this band, where we align the grid to the right
singular vectors of the Jacobian of \(\bm{f}\) at the maximum
likelihood fit.  The grid spacings in each direction are the desired
precision divided by the corresponding singular values of the
Jacobian.

For each candidate pair \(\theta\), we compute \(\bm{f}(\theta)\) to
determine a \(Q\) that is gauge equivalent to \(Y\).  Write
\(Q=W\Lambda W^{-1}\) with \(W\) the matrix whose columns are the
right eigenvectors of \(Q\). The similarity
transformation that takes \(Y\) to \(Q\) is \(WV^{-1}\). The
associated values of \(\pi\) and \(M\) satisfy the equation
\(MQ^{n}\pi = Y^{n}A\) for all \(n\). Rewriting the equation gives
\begin{align}
  M W\Lambda^{n}W^{-1}\pi = V\Lambda^{n}V^{-1}A.
  \label{eq:mn_gauge}
\end{align}
We require that \(\{Y^{k}A\}_{k}\) and therefore
\(\{\Lambda^{k}V^{-1}A\}_{k}\) are spanning and show that \(V^{-1}MW\)
commutes with \(\Lambda\).  It follows from these equations that
\(\qty(V^{-1}MW) \Lambda^{k}W^{-1}\pi= \Lambda^{k}V^{-1}A\) for all
\(k\), so
\begin{align}
  \Lambda^{n-k}\;\Lambda^{k}V^{-1}A
  &= \Lambda^{n} V^{-1}A\notag\\
  &=
    \qty(V^{-1}MW)\Lambda^{n}W^{-1}\pi\notag\\
  &=\qty(V^{-1}MW)\Lambda^{n-k}\qty(V^{-1}MW)^{-1}\qty(V^{-1}MW)\Lambda^{k}W^{-1}\pi
    \notag\\
  &=\qty(V^{-1}MW)\Lambda^{n-k}\qty(V^{-1}MW)^{-1}\;\Lambda^{k}V^{-1}A.
\end{align}
Let \(k=n-1\) and use arbitrariness of \(n\) to conclude that
\(\Lambda=\qty(V^{-1}MW)\Lambda\qty(V^{-1}MW)^{-1}\), and consequently
\(V^{-1}MW\) commutes with \(\Lambda\).

The commutativity established in the previous paragraph implies that
\(M=V\Sigma W^{-1}\) for a real diagonal matrix \(\Sigma\) with
diagonal \(\bm{\sigma}=(\sigma_{1},\ldots,\sigma_{4})\).  Thus
  \(M\) is a linear function of \(\Sigma\).  From
Eq.~\eqref{eq:mn_gauge} for \(n=0\), we see that \(M\pi=A\) and
therefore \(\pi=W\Sigma^{-1} V^{-1}A\).  To determine whether
\((l_{g},s_{g})\) is feasible we need to determine whether there
exists \(\Sigma\) such that \(M\) and \(\pi\) satisfy the
constraints. Since \(M\) is linear in \(\bm{\sigma}\), and the
constraints on \(M\) are linear, we obtain a set of linear constraints
on \(\bm{\sigma}\) that ensure that \(M\) satisfies its
constraints. Let \(C_{M}\) be the convex set consisting of the
\(\bm{\sigma}\) satisfying these constraints. For a pattern of four
signs \(\bm{s}=(s_{1},\ldots,s_{4})=(\pm1,\pm1,\pm1,\pm1)\) let
\(C_{M,\bm{s}}\) be the intersection of \(C_{M}\) with the orthant
where the entries of \(\bm{\sigma}\) have the given signs.

Consider a given set \(C_{M,\bm{s}}\) of candidate feasible values of
\(\bm{\sigma}\).  Necessary conditions for satisfying the constraints
on \(\pi\) can be obtained as follows. From
\(\Sigma^{-1} V^{-1}A=W^{-1}\pi\) we deduce that
\(\sigma_{i}=(V^{-1}A)_{i}/(W^{-1}\pi)_{i}\). Let \(s'_{i}\) be the
sign of \((V^{-1}A)_{i}\).  The constraints on \(\pi\) imply readily
obtained upper and lower bounds on the possible values of
\((W^{-1}\pi)_{i}\) with the sign \(s'_{i}s_{i}\). These bounds imply
upper and lower bounds on \(\sigma_{i}\) in the orthant under
consideration. If the four thus obtained upper and lower bounds can be
satisfied in at least one \(C_{M,\bm{s}}\), we declare
\((l_{g},s_{g})\) feasible.  Feasibility is a relaxation of
validity. This relaxation is an adaptation of the McCormick
envelope~\cite{McCormick1976}.

\section{Experimental results and comparisons}
\label{sec:exresults}

\subsection{DSRB results}
\label{sec:exresults:dsrb}

The design of an experiment requires that we choose a family
\(\{(s_{j},m_{j}, n_{j})\}_{j=1}^{d}\) of configurations consisting of
preparations, measurements and sequence lengths , where \(n_{j}\)
trials are performed at configuration \(j\). For the dual-rail qubit
experiment, we used only one preparation, so configurations need only
specify the measurement and sequence length.  The design used in our
experiment is described in App.~\ref{sec:coptimaldesign}, where we
give the fraction of trials performed at each configuration.  The
total number of trials was \(4\times 10^{6}\), and they were assigned
to each configuration according to their design fraction, truncating
to integers as needed. Ideally, the sequence for each trial is chosen
randomly and independently according to the theory in
Sect.~\ref{sec:decoheredrb}. Full randomization is required for
correctness of the likelihood function and has the benefit of
eliminating added variance from running the same sequence more than
once~\cite{kwiatkowskiThesis2024,kwiatkowski2023optimized}.  However,
our control hardware did not permit such full randomization. Instead,
for each configuration we chose \(40\) random sequences, where we ran
\(1/40\) of the trials assigned to the configuration in one block of
consecutive trials for each chosen sequence. We used the same \(40\)
sequences for the two measurement settings at a given sequence length,
when both measurement settings had non-zero fractions of trials
assigned to them.  The choice of final permutation in the sequence was
also random and instead of recording the choice, the outcome was
labeled by whether it was as expected for the ideal sequence
(``correct'') or not (``incorrect'').  This has the effect of
symmetrizing the measurement matrices as explained at the end of
Sect.~\ref{sec:stochasticmodel}.

When the experiment is not fully randomized, the distribution of
outcome probabilities over sequences contributes extra variance and
requires changing the likelihood functions in the maximum likelihood
algorithm.  Since the sequence distribution is unknown, the likelihood
function is not determined. One can use the empirical covariance
matrices for the sequence-dependent outcome frequencies with a
Gaussian heuristic to change the problem to one of applying a weighted
least-squares fit. This method is effective if the number of trials
for each configuration is reasonably large and the outcome frequencies
are well inside the probability simplex.  We used this method for
applying the comparison leakage benchmarking methods to the data from
configurations involving measurement \(M_{2}\), see below. It is
possible to apply this method also for fitting the Markovian model. To
determine the effect on the estimates from changing the objective
function in maximum likelihood to weighted least squares, we
implemented weighted least squares, initializing the optimizer with
the maximum likelihood solution found by expectation maximization.
While the endpoints of the confidence intervals including gauge shift
by somewhat more than expected based on the bootstrap distributions,
the intervals have significant overlap. For example, the conditional
error and total leakage parameters from the weighted least squares fit
and from maximum likelihood satisfy
\(\epsilon\in[0.0006,0.0011]\) 
and \(l_{g}+l_{e}\in[0.0102,0.0110]\) compared to
\(\epsilon\in[\conderrl,\conderrh]\) and
\(l_{g}+l_{e}\in[\totalleakl,\totalleakh]\) for the ML fit.  See the
violin-plot histograms of gauge interval endpoints in
Fig.~\ref{fig:result} for the endpoint histograms.  Because of this
overlap and the preference for full randomization, we are presenting
the results from the maximum likelihood analysis of
Sect.~\ref{sec:inference}.  However, we account for the added variance
from not fully randomizing by modifying the bootstrap resampling
method when determining the bootstrap confidence intervals, see below.

The results from analyzing the experimental data are shown in
Fig.~\ref{fig:result}.  The first step of the analysis is to obtain a
maximum likelihood estimate satisfying the SPAM constraints by
applying the algorithm described in App.~\ref{app:em}.  The
expectation maximization algorithm returns a critical point of the
likelihood function. To help find the parameters maximizing the
likelihood, we used four random starting points obtained from random
bounded perturbations of the reference point of
App.~\ref{sec:coptimaldesign} ``projected'' back into preparation,
measurement and transition matrices satisfying the constraints. See
Sect.~\ref{sec:notes} for more details on implementation.

The maximum likelihood estimate is determined only up to gauge, so the
second step is to determine a good bounding region of the parameters
\(\bm{f}\) of the transition matrix by applying the algorithm
described in App.~\ref{app:gauge search}.  We project this region onto
each of the coordinates of \(\bm{f}\) to obtain point estimates of the
upper and lower bounds for each of the parameters
\(\epsilon, l_{g},l_{e},s_{g},s_{e}\). The figure shows confidence
intervals for the parameters, which are obtained by non-parametric
bias-corrected bootstrap as explained in the next paragraph.

The analysis for the upper and lower bound point estimates is
performed as if the experiment were fully randomized. The extra
variance from not fully randomizing is taken into account by the
bootstrap. For this purpose, we non-parametrically bootstrap resample
the data \(1000\) times with the following procedure. For each
configuration, we sample with replacement \(40\) times from the
\(40\) sequences used. For each chosen sequence, let \(m\) be the
number of trials for this sequence. We randomly sample \(m\)
outcomes from the empirical frequencies of the observed outcomes.
We then apply the steps used to obtain the point estimates to the
bootstrap resample. This bootstrapping procedure yields \(1000\)
bootstrap upper and lower bounds for each of the parameters. The
distributions of the bounds obtained are shown as violin-plot
histograms in Fig.~\ref{fig:result}.  For the final confidence
intervals shown in the table of the figure, we used bias-corrected
bootstrapping~\cite{efron_bootstrap,efron_jackknife} for one-sided
confidence bounds for each end of the resampled interval estimates.

As an indication of what preparation vector and measurement matrices
are found by the ML algorithm, their point estimates are given by
\begin{align}
  \pi_{\ML} &=
              \kbordermatrix{&\\
  0& 0.9567\\
  1& 0.0079\\
  g& 0.0315\\
  e& 0.0038\\
  },  \label{eq:pi_mlpoint}
\end{align}
\begin{align}
  M_{1,\ML}
  &=
    \kbordermatrix{
  &0&1&g&e\\
  0& 0.0437& 0.0437& 0.9868& 0.2000\\
  1& 0.0072& 0.0072& 0.0132& 0.8000\\
  2& 0.9490& 0.9490& 0.0000& 0.0000
                             },
                             \notag\\
  M_{2,\ML}
  &=
    \kbordermatrix{
  &0&1&g&e\\
  0& 0.9326& 0.0097& 0.0296& 0.1000\\
  1& 0.0097& 0.9326& 0.0296& 0.1000\\
  2& 0.0576& 0.0576& 0.9408& 0.8000
                             }.           \label{eq:M_MLpoint}
\end{align}
These estimates are representative of the ML gauge and are therefore
interpretable only up to gauge freedom. They satisfy the constraints
by design. The only saturated constraints are the ones related to the
outcome probabilities if the state before measurement is \(\ket{ee}\).
This is likely due to the fact that because of the small value of
\(l_{e}\), the population in \(\ket{ee}\) in the experiment must have
been close to undetectable. This implies there is very little basis
for inferring these outcome probabilities. We also observe that the
probability of outcome \(2\) with measurement \(M_{2,\ML}\) for a
computational state is of the order of \(6\;\%\). The gauge interval
for this probability is \([0.051,0.066]\).

\begin{figure}[h]

  \includegraphics[width=\linewidth]{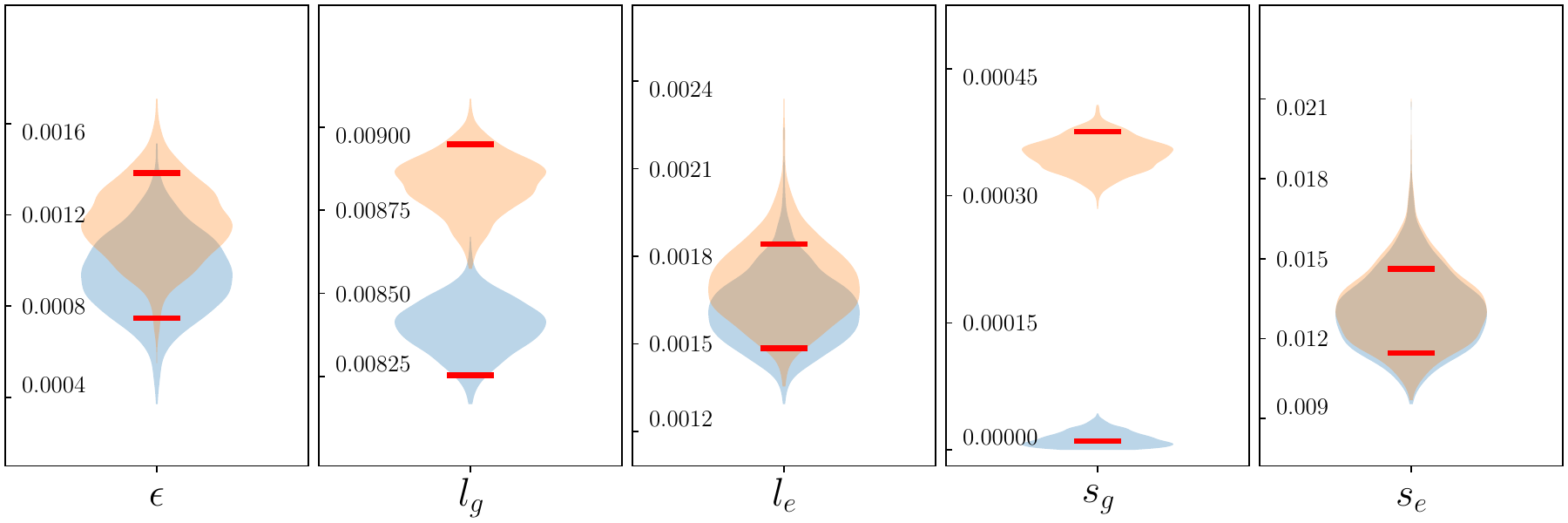}

  \vspace{1em}

  \begin{center}
  \begin{tabular}{cc}
    \hline\hline
    Parameter & \(68\;\%\) confidence Interval\\
    \hline
    $\epsilon$ & $[0.00075, 0.00138]$ \\
    $l_g$      & $[0.00825, 0.00895]$ \\
    $l_e$      & $[0.00148, 0.00184]$ \\
    $s_g$      & $[0.00001, 0.00038]$ \\
    $s_e$      & $[0.01146, 0.01462]$ \\
    \hline\hline
  \end{tabular}
  \end{center}

  \caption{Analysis results from 1000 nonparametric bootstraps.  The
    subplots show the bias-corrected \(68\;\%\) confidence intervals
    (red horizontal lines) and violin-plot bootstrap histograms for
    the bootstrap upper (orange histograms) and lower bounds (blue
    histograms) for each parameter of the transition matrix. The
    numerical confidence intervals are listed in the table.  }
  \label{fig:result}
\end{figure}

\subsection{Comparison to conventional leakage benchmarking}

For comparison we consider the three alternative methods referenced in
the introduction. Unlike DSRB, these methods infer the total leakage
rate \(l_{g}+l_{e}\) and the error \(\epsilon\) conditional on not
having leaked. They make assumptions that eliminate the gauge freedom.
They use only one measurement setting, corresponding to measurement
\(M_{2}\), so for comparison, we consider only the configurations with
this measurement.  For these configurations, the measurement outcomes
are classified as ``correct'', ``incorrect'' or ``leaked'', labeled
with \(0,1\) or \(2\) in Eqs.~\eqref{eq:ideal_measurements}
and~\eqref{eq:M_MLpoint}.  Each method involves a model for the
probabilities of some or all of these outcomes as a function of
sequence length and is based on fits to the model.  For our
comparison, all fits are based on weighted least squares with weights
determined by the empirical variance of the relevant combination of
outcome frequencies over sequences at each configuration. The variance
is determined from the empirical covariance matrix for the outcome
frequency vectors and accounts for the lack of full sequence
randomization.  This is motivated by the fact that the number of
trials at each sequence with measurement \(M_{2}\) is large, at least
\(7,578\) so that the binomial sampling variance at a given sequence
is negligible compared to the sequence variance. Except in one case
(see below), the fits exclude the length \(0\) sequences, both because
they did not contribute to DSRB and because they significantly
worsened the fit quality in several cases.

The first step of each method is to infer the total leakage rate
\(l_{g}+l_{e}\) and if possible the total seepage rate \(s_{g}+s_{e}\)
from the observed frequencies of not having leaked as a function of
sequence length. This step involves a fit to an exponential decay or a
linear fit for short sequence lengths. 
In the second step, the methods fit a model to the observed
frequencies of ``correct'' and ``incorrect'' outcomes as a function of
sequence length. The average conditional gate error \(\epsilon\) is
determined from the parameters of the model used.

The Wood-Gambetta
method (WG)~\cite{woodQuantificationCharacterizationLeakage2018} assumes
that the leakage subspace is depolarized and dephased from the
computational subspace in each step of the randomized sequence.  In
the Markov model it suffices to assume that the two leakage levels
behave identically, as if depolarized, or only one leakage level
participates significantly. The gauge problem is avoided by assuming
that there are no measurement errors for the three outcome
measurement. The DSRB estimates in Fig.~\ref{fig:result} suggest that
these assumptions are significantly violated.  The protocol is
described in Sect. III of the reference. In the first step, total
leakage and seepage is determined by fitting an exponential decay to the
observed frequencies of ``leaked'' outcomes. The fitting function has the form
\(A+B\lambda_{1}^{n}\).  From the parameters of the fit one infers
the total leakage and total seepage. In the next step, the frequencies
of ``correct'' outcomes are fitted to a model with two exponential decays,
where one of the decay parameters is fixed by the first fit.  The
conditional error \(\epsilon\) is then determined by the fit
parameters from both fits.

For the postselected method (Post) we use an extension of leakage randomized
benchmarking as described in Ref.~\cite{hashim2025practical}
Sect. VIII.I.1.  Various versions of postselected methods are used in
the references cited in the introduction. We use the version described
next.  The total leakage rate and seepage rate (if desired) is
determined as in the first step of the Wood-Gambetta method and makes
the same assumptions. The conditional error is obtained from the
decay parameter of a fit to \(C+D\lambda_{3}^{n}\) of the ratio of the
frequencies of ``correct'' outcomes to the sum of the frequencies of
``correct'' and ``incorrect'' outcomes. Validity of the conditional error
inference requires no seepage and no error in identifying the
computational subspace in the measurement as well as decoherence
between the computational and leakage subspaces. The weights for this
fit are obtained from the sample variance of the ratios being fitted
over the \(40\) random sequences. The
inferred functional model for the probabilities of ``correct'' outcomes is then
\((1-A-B\lambda_{1}^{n})(C+D\lambda_{3}^{n})\).

There are a number of possible Chen-Baldwin (CB) protocols and their
combinations that can be applied. The possibilities are summarized in
Ref.~\cite{ChenRBLeakage2025} Table I. Here we combine a version of
leakage estimation under the short lengths assumption with an estimate
of the conditional error according to the average over measurement
basis (Avg.~MB) protocol under the population transfer assumption.
The short lengths assumption is valid when the leakage rate is small
compared to the inverse of the sequence lengths used, and the
probability of misidentifying the leaked subspace in the measurement
is negligible.  The population transfer assumption imposes dephasing
between the computational and the leakage subspace and between leaked
levels in each step. Additional assumptions on the measurement and
the gates are required, both to avoid the gauge problem and to ensure
a simplification of the averaged measurement operators, see Eq.~(16)
in the reference.

Total leakage is determined by fitting a straight line to the
frequencies of ``leaked'' outcomes for the short sequence lengths.  We use only
the first two available lengths, namely length \(0\) and length
\(14\). The next available length is \(80\) and does not satisfy the
short lengths assumption.  For the fit to total leakage we modify
Eq.~(29) in the reference by allowing an arbitrary offset to account
for an initial leaked population.  The total leakage rate is
determined by fitting a line to the observed frequencies of
``computational'' (``correct'' or ``incorrect'') outcomes for short sequence
lengths.  Because there are only two sequence lengths, the line is
simply determined by the frequencies at length \(0\) and \(14\).  The
slope of this fit yields an estimate of the total leakage rate
$l_{g}+l_{e}$.  Because the other lengths did not satisfy the short
lengths assumption, we used only the lengths \(0\) and \(14\) in the
experimental design for this estimate. This method for estimating
total leakage does not yield a complete model for the frequencies of
``leaked'' outcomes.
  
To determine the conditional error \(\epsilon\) we used the Avg.~MB
protocol.  This protocol requires averaging over the final permutation
in a sequence before a \(d\)-outcome measurement. See the description
in the paragraphs after Eq.~(13) in Ref.~\cite{ChenRBLeakage2025}.
The recorded data does not contain the exact information required by
the Avg.~MB protocol. However, under the assumption that the
measurement has no error and in view of our randomization of the final
permutation, the Avg. MB protocol requires that we fit an exponential
decay \(D'\lambda_{4}^{l}\) to the difference between the frequencies of
``correct'' and ``incorrect'' outcomes, where we modified the protocol slightly
by fitting both the amplitude and \(\lambda_{4}\).  From the decay
rate of the fit one obtains an estimate of the eigenvalue
\begin{align}
  \lambda_1 = (1-l_{g}-l_{e})(1-2\epsilon)
\end{align}
of the Markov process. 
We can solve for \(\epsilon\) given the
estimate of \(l_{g}+l_{e}\) from the total leakage determination.
In the absence of measurement error and given the true value
of total leakage, this is a consistent estimate of \(\epsilon\) according
to the Markov model.

The values of total leakage and conditional error obtained with the
three alternative methods are shown and compared to those obtained
with DSRB in Fig.~\ref{fig:comparison} (a). Notably, the conditional
error is biased high, while the total leakage is biased low, compared
to DSRB. The biases are significant in some cases, but are subject to
uncharacterized biases from model violation.  None of the alternatives
account for the gauge problem, which accounts for a notable fraction of
the confidence interval for DSRB.

\begin{figure}[h]
  \centering
  \includegraphics[width=0.7\columnwidth]{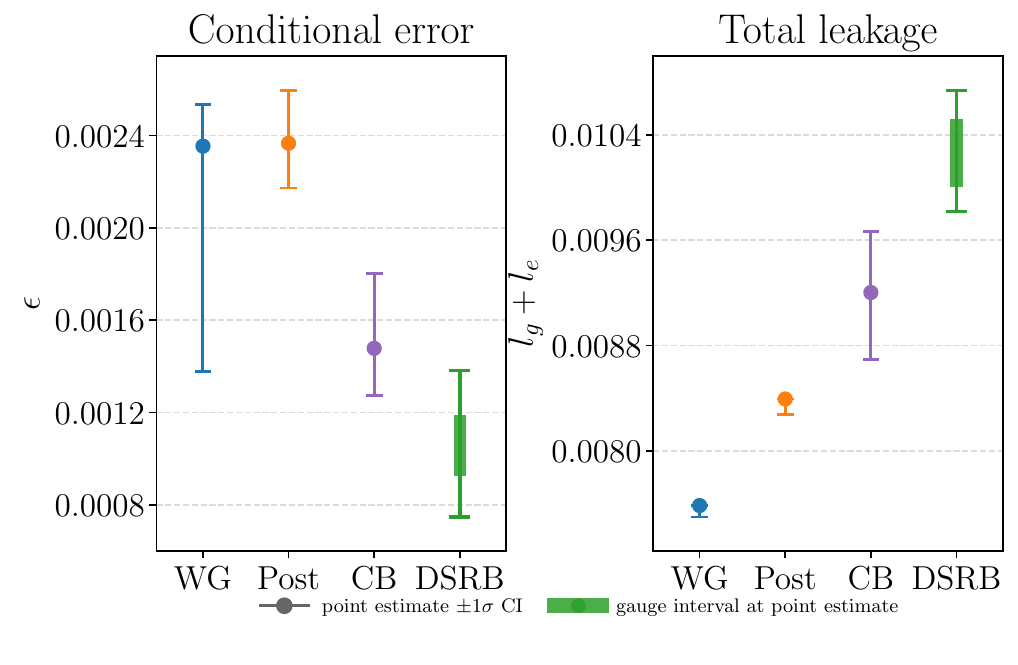}
  \caption{Estimates of the conditional error \(\epsilon\) and total
    leakage \(l_{g}+l_{e}\) according to the three alternative methods
    compared to DSRB.  The estimates for WG, Post and CB are based on
    the configurations with measurement \(M_{2}\) only.}
  \label{fig:comparison}
\end{figure}

The observed biases could be due to poor fits, for example by failing
to account for contributions from sufficiently many exponential decays
that can arise from having different leakage and seepage rates to the
two leaked levels.  We therefore performed a comparison of fit
quality, with fit quality based on the sample-covariance-weighted
mean-square difference with the empirical covariance matrix for the
outcome frequency vectors. This is also called the Hotelling's \(T^2\)
statistic and can be viewed as a version of the usual \(\chi^{2}\)
statistic that accounts for finite sampling effects.  We considered
the \(T^{2}\) for fits to functions of the observed frequencies,
specifically the difference between the frequencies of ``correct'' and
``incorrect'' outcomes (``correct''\(-\)``incorrect''), the frequencies of
``leaked'' outcomes, and the frequencies of ``correct'' outcomes.  To compare to DSRB,
we also calculated \(T^{2}\) for the DSRB point estimate. The gauge
freedom of the point estimate does not affect the value of this
statistic.  DSRB is based on maximum likelihood and therefore
optimizes a different statistic with data that includes the other
measurement setting. The fit quality for the DSRB point estimate is
therefore not optimized over Markov model parameters, which may affect
the comparison.  The \(T^{2}\) statistics obtained for DSRB and the
three alternatives are shown in Table~\ref{table:T2}, with values
itemized by sequence length.  Of the three fits considered, CB only
models ``correct''\(-\)``incorrect'', while WG and Post model all
three.

The fact that CB directly fits the frequency difference is reflected
in a lower \(T^{2}\). DSRB and the Wood-Gambetta method have similar
\(T^{2}\), but the postselected method performs worse, to a large
extent because of a poor match at sequence length \(14\).  The DSRB
fits are overall better than those of WG and Post, except that CB is a
good match to ``correct''\(-\)``incorrect'', as might be expected from
the Markov model.

\begin{table}[htbp]
  \centering
  \caption{Table of \(T^{2}\) statistic fits for configurations with
    measurement \(M_{2}\) for DSRB and the three alternative
    methods. All alternative methods provide fits to
    ``correct''\(-\)``incorrect''. WG and Post also provide fits to
    ``leaked'' and to ``correct''. For ``leaked'' Post uses the same
    method as WG.}
  \label{table:T2}
  \begin{tabular}{|r|rrrr|rr|rrr|}
    \hline
    &\multicolumn{4}{c|}{``correct''\(-\)``incorrect''}&
    \multicolumn{2}{c|}{``leaked''}&
    \multicolumn{3}{c|}{``correct''}\\
    \(n\) & DSRB & WG & Post & CB& DSRB & WG/Post&DSRB & WG & Post\\[0pt]
    \hline
    \(14\) & \(0.24\) & \(0.01\) & \(31.32\) & \(0.00\)& \(0.33\) & \(34.80\)& \(0.39\) & \(13.34\) & \(44.99\)\\[0pt]
    \(80\) & \(0.05\) & \(0.83\) & \(1.98\) & \(0.12\)& \(13.54\) & \(39.74\)& \(1.00\) & \(5.15\) & \(0.05\)\\[0pt]
    \(227\) & \(0.96\) & \(0.09\) & \(0.06\) & \(0.14\)& \(3.85\) & \(0.47\)& \(0.31\) & \(0.17\) & \(0.12\)\\[0pt]
    \(228\) & \(0.52\) & \(0.69\) & \(0.59\) & \(0.00\)& \(27.77\) & \(40.06\)& \(4.71\) & \(1.08\) & \(1.20\)\\[0pt]
    \(599\) & \(0.02\) & \(0.01\) & \(0.83\) & \(0.02\)& \(21.45\) & \(2.96\)& \(1.19\) & \(0.06\) & \(0.37\)\\[0pt]
    \hline
    total & \(1.78\) & \(1.63\) & \(34.78\) & \(0.29\)& \(66.94\) & \(118.03\)& \(7.60\) & \(19.79\) & \(46.73\)\\[0pt]
    \hline
  \end{tabular}
\end{table}

\begin{figure}[h]
\centering
\includegraphics[width=0.6\columnwidth]{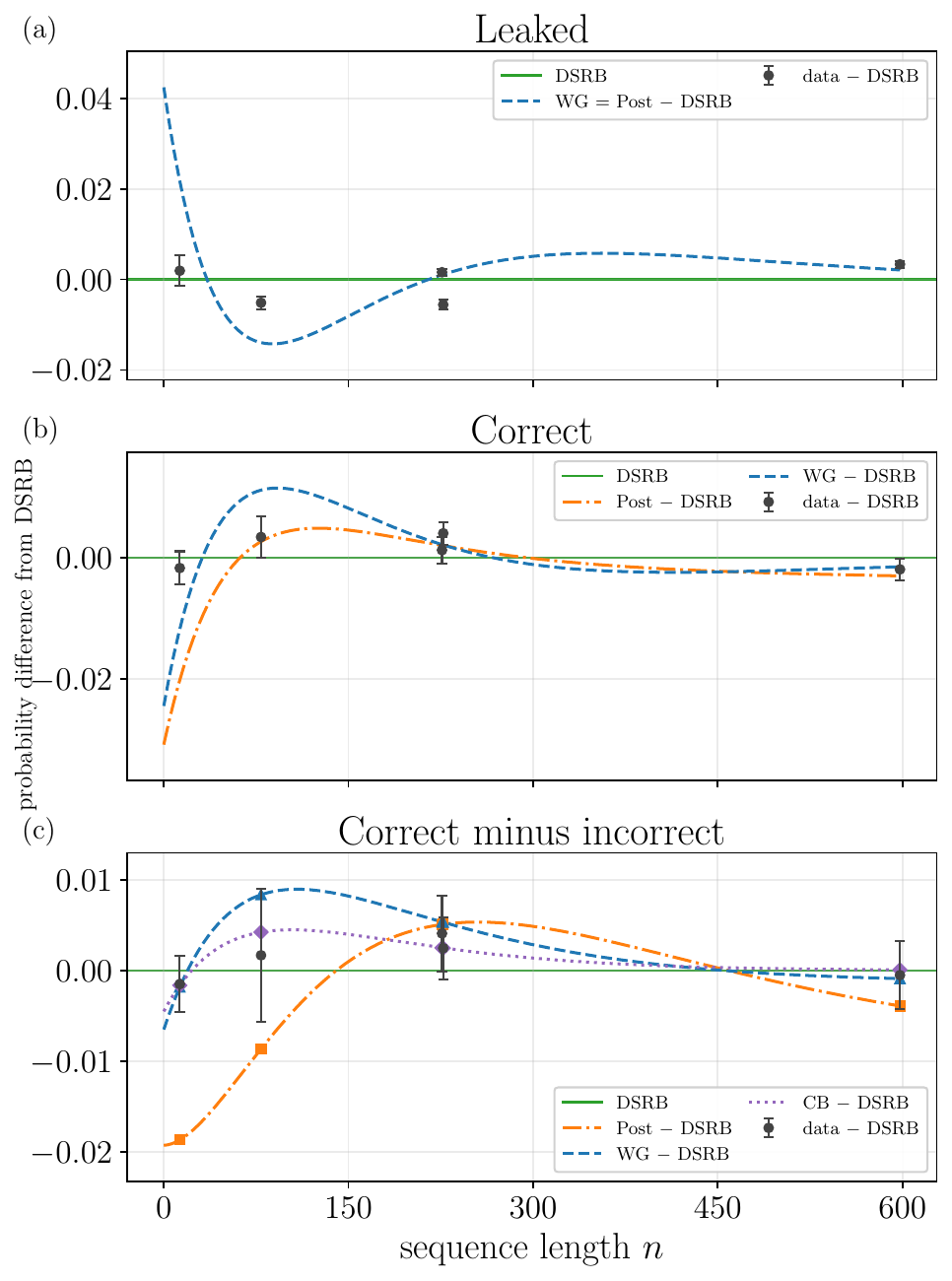}
\caption{Visualizations of the fits by DSRB and the comparison
  methods.  Each subplot shows
  the frequencies according to the data with
  error bars of one standard deviation together with the probabilities
  according to the relevant fits. To highlight the differences to DSRB,
  the DSRB probabilities are subtracted. Subplot (a) shows
  the probabilities of ``leaked''. In this case WG and Post use the same
  fitting procedure. CB does not have a complete model for these probabilities
  and is not shown. Subplot (b) shows the probabilities of ``correct'' for
  WG and Post. Again, CB does not have a complete model for these probabilities.
  Subplot (c) shows the difference probabilities ``correct''\(-\)``incorrect'',
  and all methods are applicable.}
\label{fig:comparison residuals}
\end{figure}

We also checked the \(T^{2}\) for the fit of DSRB to all observed
frequencies at all configurations. The values of \(T^{2}\),
itemized by configuration are
\begin{center}
  \begin{tabular}{|lr|}
    \hline
    Measurement setting & \(T^{2}\)\\
    \hline
    Measure \(M_{2}\) (``resolve computational''), \(5\) lengths & \(77.27\)\\
    Measure \(M_{1}\) (``resolve leaked''), \(9\) lengths & \(235.40\)\\
    \hline
    total & \(312.67\)\\
    \hline
  \end{tabular}
\end{center}
For comparison, if the observed frequencies are sampled from the
Markov model with \(40\) sequences at each configuration with unknown
sequence covariance for the three outcomes, the expected value of the
total \(T^{2}\) with respect to this model is \(30.33\) with a
standard deviation of \(8.56\) based on Monte Carlo sampling.  The
total \(T^{2}\) with respect to a fit is smaller. Therefore, the large
value determined above suggests that the experiment deviated from the
Markov model. Such deviations can cause difficult to assess systematic
biases in the estimated quantities.

\subsection{Experiment and analysis history}

For context we describe the history of the experiment and its
analysis. The experiment was designed and implemented in mid 2024
while the theory was still under construction.  The data from the
experiment was held back while analysis methods and code were in
development. We did not have sufficient data to divide it into
practice data and data to apply a final analysis to.  Analysis methods
were first tested on simulations. However, real data often amplifies
issues that are not apparent in simulations, and the data from this
experiment was no exception. Consistency checks on the results from
our first analysis attempts in 2025 revealed identifiable problems
that required correction.  In early 2026, competent AI coding
assistance became available.  AI help was used to correct for the
initial issues found and dramatically improve the efficiency of the
analysis code, in part by drawing on advanced results from the
literature outside our fields.  In the meantime, the fully developed
theory showed that there were additional constraints on measurement
matrices arising from the relationship between the outcome probability
distributions for different final permutations. We eventually
discovered that these constraints are not automatically discovered by
maximum likelihood, making it necessary to explicitly respect them in
inference. In addition, the experimental design included trials at
sequence length \(0\), which are not accounted for by the Markov model
and therefore needed to be excluded.  The experiment and its analysis
were very much a learning experience, and ideally we would update the
experimental design and perform a confirmatory experiment. However,
the experimental system was retired in mid 2025 with no replacement or
an upgraded version currently available.

Whenever analysis of experimental data undergoes multiple revisions,
it is possible to introduce biases in the final results.  The
possibility of such biases can be revealed by the history of the
reported values over analysis attempts that subsequently had to be
corrected. In our case, there were two cycles of analysis before we
used AI assistance to reimplement and improve the analysis code and
three major cycles afterwards.  The last three relevant inferred gauge
intervals for conditional error and total leakage were, in order
\begin{center}
  \begin{tabular}{cccc}
    \hline
    parameter & first & second & last\\
    \hline\hline
    \(\epsilon\)&[0.0016,0.0023]&[0.0012,0.0016]&[\conderrl,\conderrh]\\
    \(l_{g}+l_{e}\)&[0.0098,0.0106]&[0.0099,0.0104]&[\totalleakl,\totalleakh]\\
    \hline
  \end{tabular}
\end{center}
The second analysis added the measurement matrix identity constraints
of Eq.~\eqref{eq:mr_twirl} and~\eqref{eq:mr_twirl_1}. The third
analysis corrected an offset of one in translating sequence length to
number of Markov steps and dropped inclusion of the length \(0\)
sequences.  While the conditional error shows a trend of improvement,
the total leakage is stable. We believe that the differences are
likely of the order of systematics introduced by not fully randomizing
and other out-of-model effects witnessed by the large total \(T^{2}\)
for the DSRB fit.

\Ac{We did not perform a specific investigation to see how the gauge
  intervals depend on the SPAM quality constraints. If either measurement
  or preparations are precisely known, there is no gauge. If the quality
  constraints are relaxed too much, the labels of the leaked states
  become arbitrary.}

\section{Discussion}
\label{sec:discussion}

The \(T^{2}\) for our model and visual inspection of the fits suggests
that the experimentally observed frequencies exhibit features not
captured by the dual-rail qubit Markov model. Some of these features
could come from not fully randomizing. However, since changing the
objective function from the fully randomized likelihood to the square
deviation according to the empirical covariance matrices has no
significant effect on the results, it is likely that the features are
due to other unmodelled effects.  Known such effects include
gate-dependent errors, incomplete decoherence and coherent leakage,
direct transitions between leaked states, state decay during
measurement, and experimental drifts.  Nevertheless, our analysis of
the fits shows that future experiments would benefit greatly from full
randomization of the sequences to prevent the added variance from
using the same sequence many times and to ensure correctness of the
likelihood function used in maximum likelihood inference. In future
experiments, we will likely revisit our methods for optimizing the
experimental design and include additional configurations for better
diagnostics of model discrepancies.

Gauge degeneracy is a common problem in state and process inference in
quantum experiments, most notably in gate-set tomography
(GST)~\cite{nielsen2021gate}. In GST, gauge degeneracy is removed by
minimizing the distance to the user's specified ideal state
preparations, gates and measurements.  In our application, the number
of parameters that need to be inferred is greatly reduced by the
effect of twirling due to the random sequences.  It is nevertheless
plausible that applying well-supported constraints on the performance
of GST quantum operations can be used to constrain inferred quantities
in an experimentally meaningful way. A number of such constraints have
been shown to be advantageous in
GST~\cite{moueddene2022contextaware,PRXQuantum.4.010325,vinas2025microscopic}.

The theory of decohering subgroup randomized benchmarking is general and
can be applied to most randomized benchmarking scenarios that can
implement gates that are close to implementing a subspace
representation of a group with a subgroup that decoheres an orthogonal
basis.  The theory shows that the observed frequencies can be modeled
by a classical Markov chain.  Here we only considered gate-independent
errors, where the gates are members of the group. If the members of
the group require composing gates from a more basic set, errors
usually depend on the group element, even if the basic gates have
gate-independent noise. The effect of this and other gate dependence
needs to be accounted for.  We expect that the theory can be extended
to gate-dependent errors, as has been done for many randomized
benchmarking scenarios~\cite{proctor2017randomized,wallman2018randomized,merkel2021randomized,helsen2022generalrb,Claes2021}. Another
extension worth considering is to add interleaved
gates~\cite{magesan2012efficient,gaebler2012randomized,chasseurCompleteRandomizedBenchmarking2015,wu2024leakage} to the random gate sequences.

Our algorithms and their implementation show that the Markov chain
model can be effectively inferred from data, and the region of
gauge-equivalent model parameters can be effectively determined or at
least bounded in the presence of box-constraints on SPAM parameters as
well as physical constraints on transitions. Our implementation is
specific to the dual-rail qubit in our experiment. The algorithms
generalize to other situations involving more leakage levels or
qubits. The algorithms are practical as long as the number of free
parameters after twirling is sufficiently small.

\begin{acknowledgments}
  \label{ack}
  This work includes contributions of the National Institute of
  Standards and Technology, which are not subject to
  U.S. copyright. The U.S. Government is authorized to reproduce and
  distribute reprints for governmental purposes notwithstanding any
  copyright annotation thereon. The use of trade names and software
  does not imply endorsement by the US government, nor does it imply
  these are necessarily the best available for the purpose used
  here. This document has not been peer reviewed but has been cleared
  by NIST for release.

  Code and data for this work is available for download
  from Ref.~\cite{geller2026leakage}.

  Use of AI: The leakage benchmarking theory was developed without the
  use of AI. In the main text, AI was used to assist with production
  of schematics and minor pre-submission editorial assistance. The
  experiment and the original versions of the analysis methods and
  code were implemented in early 2025 without the use of AI.  Once
  sufficiently competent AIs became available in 2026, the analysis
  code was dramatically improved and reimplemented by AI. Much of the
  description of the specific implementation in the appendix was initially
  generated by AI. The AI generated text needed to be substantially rewritten
  for human readability.
\end{acknowledgments}

\bibliography{leakagerb,em_constrained}

\pagebreak

\appendix
\makeatletter
\renewcommand{\p@subsection}{\thesection.}
\makeatother

\section{A-optimal experimental design}
\label{sec:coptimaldesign}

For the experiment, we optimized and implemented an experimental
design that determines the fraction of the total number of trials to
run at each configuration.  The goal of the optimization is to
minimize a sum of the variances of the optimal linear estimators
associated with the parameters $\bm{f}$ of the transition matrix at a
reference point for preparations, measurements and the transition
matrix. The reference point was based on device characterization
that was performed before the experiment and is given below. Let
\(\nu_{1}(n)\) and \(\nu_{2}(n)\) be the fraction of the total number
of trials of the experiment to be performed at setting \(M_{1}\)
(``resolve leaked'') and at setting \(M_{2}\) (``resolve
computation''), respectively.  We set an upper bound of \(599\) on the
sequence length so that \(\sum_{n=0}^{599}(\nu_{1}(n)+\nu_{2}(n))=1\).

We used techniques from optimal design
theory~\cite{fedorov2010optimal} to optimize
\(\nu_{1}(n), \nu_{2}(n)\).  The version of the techniques used here
is described in
Ref.~\cite{geller2024characterizationnoninteractingbosonsapplications},
section 4.3.  Because of the gauge freedom, the parameters \(\bm{f}\)
are not uniquely inferrable. We therefore minimized the sum of the
variances of a standardized maximal set of inferrable parameters as
described below. The reference point used to optimize the experimental
design is given by
\begin{align}
  \text{\begin{tabular}{c|c|c|c|c}
          \(\epsilon\) & \(l_g\) & \(l_e\) & \(s_g\) & \(s_e\) \\
          \hline
          0.005 & 0.01 & 0.002 & 0.001 & 0.04 \\
        \end{tabular}}
\end{align}
\begin{align}
  \pi &= \begin{pmatrix}0.98 \\ 0.008 \\ 0.008 \\ 0.004\end{pmatrix},\\
  M_{1} &=  \begin{pmatrix}
    .03 & .02 & .92 & .03\\
    .02 & .03 & .05 & .94\\
    .95 & .95 & .03 & .03
  \end{pmatrix},\\
  M_{2} &=  \begin{pmatrix}
    .95 & .02 & .005 & .03\\
    .02 & .95 & .005 & .03\\
    .03 & .03 & .99  & .94
  \end{pmatrix}.
\end{align}
The measurement \(M_{1}\) does not satisfy all the constraints of
Eq.~\eqref{eq:mr_twirl_1}. This is because at the time the
experimental design was optimized, we did not yet appreciate the
consequences of randomizing the final permutation. As a result we did
not include the constraints of Eqs.~\eqref{eq:mr_twirl}
and~\eqref{eq:mr_twirl_1} in determining the sum of variances to be
optimized in the next paragraph.

Because the columns of \(\pi\), \(M_{1}\) and \(M_{2}\) sum to \(1\),
there are \(24\) independent parameters.  We considered sequence
lengths between \(0\) and \(599\). We combined \(\nu_{1}\) and
\(\nu_{2}\) into one vector \(\nu\) of length \(1200\). Because of the
sum constraint, \(\nu\) has \(1199\) independent variables.  To
determine an effective sum of variances to minimize, it is necessary
to compute the Jacobian at the reference point of the map from the
\(24\) independent parameters determining the model to the independent
observable probabilities at the lengths and measurements under
consideration.  Because of gauge, this Jacobian is not full rank. The
gauge-free representation used in Sect.~\ref{sec:gaugeopt} has \(19\)
free parameters: Because one of the eigenvalues of \(Y\) is \(1\),
\(Y\) has at most \(15\) free parameters, while \(A\) has at most
\(4\) free parameters. To confirm independence of these parameters, we
verified numerically that the rank of the Jacobian is \(19\). For the
next step, we used the Euclidean inner product with respect to the
basis for the \(24\) independent parameters obtained from the first
two rows of each \(M_{i}\), the first three entries of \(\pi\) and the
entries of \(\bm{f}\). We obtained an orthonormal basis of the
orthogonal complement of the null space of the Jacobian and
constructed the orthogonal, \(24\times 19\) matrix \(V\) whose columns
consist of the basis elements.  The matrix \(V\) can be interpreted as
an isometry that maps the \(19\) standardized inferrable parameters
into the \(24\)-dimensional parameter space, where the model has been
linearized at the reference point.  We can consider the linearized
model restricted to the \(19\) standardized parameters, where these
parameters are first mapped to the \(24\) parameters with \(V\), then
to the observable probabilities.  Because we wish to minimize the sum
of variances for the parameters of \(\bm{f}\), we define \(U\) as the
matrix consisting of the rows of \(V\) corresponding to the parameters
of \(\bm{f}\) according to our coordinatization.  The objective
function \(g\) is obtained by computing the inverse Fisher information
matrix \(F\) for the \(19\) standardized parameters and computing
a weighted combination of its diagonal entries
according to
\begin{align}
  g = \sum_{i=1}^{19}w_{i}^{2}F_{ii},
\end{align}
where \(w_{i}^{2}=|\sum_{p}U_{pi}|^{2}\).  This choice is somewhat
arbitrary and loosely motivated by the desire to have reasonable
inferrability of all parameters while prioritizing \(\bm{f}\).  The
procedure for optimizing the objective function \(g\) described in
Ref.~\cite{geller2024characterizationnoninteractingbosonsapplications}
can then be applied to optimize \(\nu\).

The optimized design is shown in Tab.~\ref{tab:optimaldesign} with a
visual representation in Fig.~\ref{fig:optimalfractions} accompanied
by the predicted probabilities of outcomes for the two measurements at
the reference point as a function of sequence length.  
\begin{table}[htbp]
\centering
\caption{Optimal experimental design fractions by sequence length.}
\begin{tabular}{c@{\hspace{0.45in}}c@{\hspace{0.45in}}c}
  \hline\hline
  Sequence&Resolve&Resolve\\
  length&computational&leaked\\
  &\(\nu_{2}\)&\(\nu_{1}\)\\
\hline
0   & 0.057 & 0.054 \\
14  & 0.161 & 0.155 \\
70  & --    & 0.071 \\
71  & --    & 0.006 \\
80  & 0.161 & --    \\
124 & --    & 0.001 \\
125 & --    & 0.064 \\
227 & 0.082 & --    \\
228 & 0.076 & --    \\
353 & --    & 0.001 \\
354 & --    & 0.013 \\
355 & --    & 0.008 \\
599 & 0.077 & 0.014 \\
\hline\hline
\end{tabular}
\label{tab:optimaldesign}
\end{table}
\begin{figure}[h]
  \centering
  \includegraphics[width=\columnwidth]{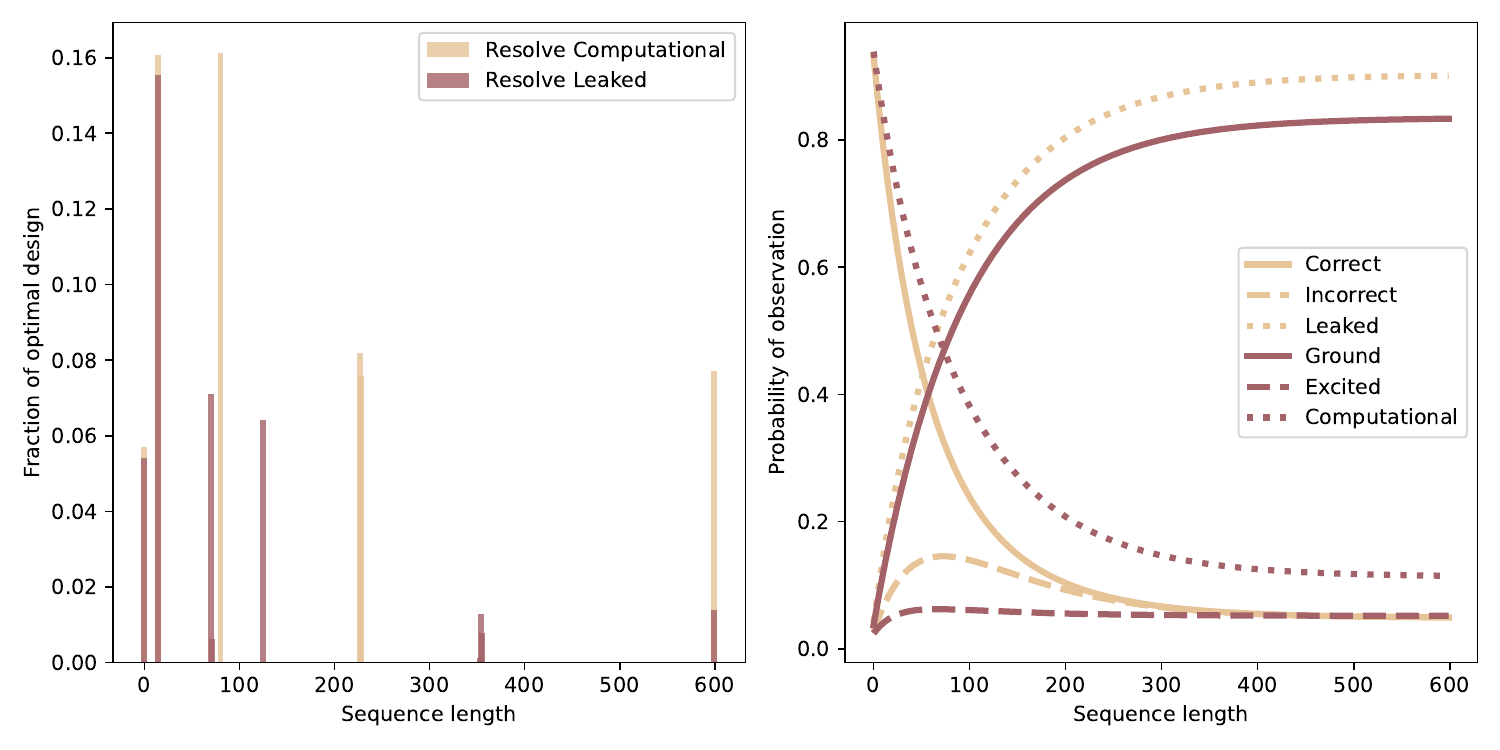}
  \caption{Experimental design. The left subfigure shows the design
    fractions at each measurement setting and length. The right
    subfigure shows the probabilities of the outcomes at the reference
    point as a function of length. ``Correct'', ``Incorrect'' and
    ``Leaked'' are the outcomes of measurement \(M_{2}\), which
    resolves the computational subspace.  ``Ground'', ``Excited'' and
    ``Computational'' are the outcomes of measurement \(M_{1}\), which
    resolves the leaked subspace.}
  \label{fig:optimalfractions}
\end{figure}

\pagebreak

\ifSubfilesClassLoaded{
  \title{Constrained EM for the Markov model}
  \author{}
  \date{\today}
  \maketitle
}{
  \section{Constrained EM for the Markov model}
  \label{app:em}
  \let\section\subsection
}

In this appendix, we present an expectation-maximization algorithm to
find a critical point of the constrained maximum-likelihood problem
for the four-state Markov model given in the main text. Critical
points are points satisfying the Karush–Kuhn–Tucker (KKT) conditions.
The critical point found is reported as the maximum likelihood
solution. We give a description of the algorithmic techniques used to
implement the expectation-maximization iteration for our
application. We then explain how we accelerate convergence of
expectation maximization by applying the SQUAREM-3 scheme of
Refs.~\cite{varadhan2008squarem,duvaradhan2020squarem}.  Since the
algorithms are specializations of well-known techniques, we rely on
the references for details. Our implementation can be found with
the code for this project at GitHub~\cite{geller2026leakage}.

Define the transition parameters
\[
  \bm\theta=(\teps,\,l_g,\,l_e,\,s_g,\,s_e),
  \qquad
  \teps\defeq\epsilon(1-l_g-l_e),
  \qquad
  \me\defeq(1-\epsilon)(1-l_g-l_e)=1-l_g-l_e-\teps,
\]
so that
\begin{align}
  Q &=
  \kbordermatrix{
    &0&1&g&e\\
    0 & \me     & \teps  & s_g/2 & s_e/2\\
    1 & \teps   & \me    & s_g/2 & s_e/2\\
    g & l_g     & l_g    & 1-s_g & 0\\
    e & l_e     & l_e    & 0     & 1-s_e
  }.
  \label{eq:Qreparam}
\end{align}
Every entry of \(Q\) is a linear function of \(\bm\theta\).  The set \(\cF_{Q}\) 
of feasible parameters $\bm{\theta}$ is defined to be the
parameters for which \(Q\) is a valid transition matrix. \(\cF_{Q}\)
is a polytope and is given by
\begin{align}
  \mathcal F_Q
  &=
  \{\bm\theta\in\rls^5\,:\,
  \teps,l_g,l_e,s_g,s_e\geq 0,\;
  \teps+l_g+l_e\leq 1,\;
  s_g\leq 1,\;s_e\leq 1\}.
  \label{eq:Qfeasible}
\end{align}

The model parameters are \(\phi=(\bm\pi,\{M_r\}_r,\bm{\theta})\),
where \(\bm\pi\) is the single preparation vector and \(\{M_r\}_r\)
are the per-setting measurement matrices.  The constraints on
\(\bm\pi\) and the \(M_r\) are column stochasticity,
\(\sum_i\pi(i)=1\) and \(\sum_o(M_r)_{oi}=1\), together with entrywise
box bounds. For the constraints used in the data analysis for the
dual-rail it suffices to add the lower bounds \(\pi(0)\geq 0.8\) and
\((M_r)_{j,k}\geq 0.8\) for \(j,k\) such that
\((\hat{M}_{r})_{j,k}=1\).  Each constraint is linear.  Let
\(\cF_{\bm{\pi}}\) and \(\cF_{\bm{M}}\) be the polytopes of parameters
\(\bm{\pi}\) and \(\{M_{r}\}_{r}\) satisfying these constraints.  The
set of parameters \(\phi\) for which \(\bm{\theta}\) is feasible and
\(\bm{\pi}\) and \(\{M_{r}\}\) satisfy their constraints is the
Cartesian product of polytopes
\[
  \mathcal F\defeq\mathcal F_{\bm\pi}\times\mathcal F_{\bm M}\times\mathcal F_Q.
\]

Each trial of the experiment has a configuration \(u=(m,n)\) that
specifies the measurement index \(m(u)=m\) and sequence length
\(n(u)=n\) for the trial. Let \(\cU\) be the set of configurations
used in the experiment, and for \(u\in\cU\), let \(y_{u}\) be the
number of trials with configuration \(u\).  A trial with configuration
\(u\) according to the Markov model with parameters \(\phi\) prepares
\(\bm\pi\), runs \(n(u)\) steps of the Markov chain, and ends with the
measurement \(M_{m(u)}\). Because we randomized the final permutation
and relabeled outcomes accordingly as accounted for by the measurement
constraints of Eqs.~\eqref{eq:mr_twirl} and ~\eqref{eq:mr_twirl_1}, we
do not need to separately account for the final permutation.  The
outcome probabilities for configuration \(u\) are
\begin{align}
  p_u(o\mid\phi)
  =
  (M_{m(u)} Q^{n(u)}\bm\pi)_o.
  \label{eq:trial-likelihood}
\end{align}
Let $y_{u, o}$ be the number of times that outcome $o$ occurred in
configuration $u$ so that \(\sum_{o}y_{u,o}=y_{u}\).  The
log-likelihood is then
\begin{align}
  \ell(\phi)
  &=
  \sum_{u\in\cU}\sum_{o=1}^O y_{u,o}\log p_u(o\mid\phi) + C_{\bm{y}},
  \label{eq:obs-ll}
\end{align}
where \(C_{\bm{y}}\) is the logarithm of a product of multinomial coefficients,
which depends only on the outcome counts \(y_{u,o}\).
To solve the maximum likelihood problem the expectation maximization algorithm
finds a critical point $\phi^\star$ for the problem
\begin{align}
  \mathop{\mathrm{argmax}}_{\phi\in\mathcal F}
  \ell(\phi).
  \label{eq:constrained-mle}
\end{align}
The solution is non-unique, because of the gauge degeneracy of the
model.  It is also possible that there is more than one connected
region of critical points. Given sufficient data and constraints that
prevent permutation of the two leaked levels, we do not expect this to
be an issue.

\subsection{Basic EM algorithm}
\label{subapp:basicem}

The general expectation-maximization (EM)
algorithm~\cite{dempster1977maximum} is a maximum-likelihood method
for models whose likelihood is a marginal over unobserved, ``hidden''
random variables. It iterates by alternating two steps referred to as
the E-step and the M-step. At each iteration it updates the current
estimate of the parameters to obtain the next estimate. The E-step
defines an auxiliary objective function of the candidate next estimate
to be maximized in the M-step.  The objective function is an
expectation with respect to the probability distribution of the hidden
random variables conditional on the observation and the current
estimate of the parameters. The expectation is of the log-likelihood
of the data and the hidden random variables with respect to the
candidate next estimate of the parameters.  This log-likelihood is
referred to as ``complete-data log-likelihood'' in the terminology
of~\cite{dempster1977maximum,mclachlan2008em}. The next estimate
consists of the parameters that maximize the auxiliary objective
function.  The iterations cannot decrease the observed-data
likelihood~\cite{dempster1977maximum, wu1983convergence}.  The method
is most useful when the maximization in the M-step has a closed form
solution. For hidden Markov models the hidden variables are the
unobserved states along the Markov chain in each instance of the
chain. The sequence of unobserved states is referred to as the latent
state sequence. For hidden Markov models, it is possible to use the
Baum-Welch
algorithm~\cite{baum1970maximization,ephraim2002hidden}. In this
algorithm, the E-step is computed by the forward-backward recursion,
and the M-step optimization problem has a closed-form solution. We use
an adaptation of the Baum-Welch algorithm that takes into account two
features of the present model: the parameters carry box constraints,
and each sequence yields a single measurement outcome at its end
rather than the conventional observation at every step.

\noindent\textbf{Complete-data log-likelihood:}  Let
$B = \{0, 1, g, e\}$ be the four states of the Markov process.  For
each individual trial, let \(\bm z=(z_0,z_1,\ldots,z_{n(u)})\in B^{n(u)+1}\)
be the latent sequence of states visited. The per-trial complete-data
log-likelihood for a trial of configuration \(u\) that produced
outcome \(o\) is
\begin{align}
  \log\Prob(o,\bm z\mid u,\phi)
  &=
  \log\pi(z_0)
  +\sum_{t=0}^{n(u)-1}\log Q_{z_{t+1},z_t}
  +\log M_{m(u)}(o\mid z_{n(u)}).
  \label{eq:per-trial-cll}
\end{align}
Write \(\bm z^{(u,o,k)}\) for the
latent sequence of the \(k\)th trial of configuration \(u\) that resulted in
outcome \(o\), with \(k\in\{1,\ldots,y_{u,o}\}\). Summing
Eq.~\eqref{eq:per-trial-cll} over all \(K\) trials gives the complete-data
log-likelihood
\begin{align}
  \log\Prob(\bm y,\bm z\mid\phi)
  &=
  \sum_{u\in\cU}\sum_{o=1}^O\sum_{k=1}^{y_{u,o}}
  \Biggl[
    \log\pi(z^{(u,o,k)}_0)
    +\sum_{t=0}^{n(u)-1}\log Q_{z^{(u,o,k)}_{t+1},z^{(u,o,k)}_t}
    +\log M_{m(u)}(o\mid z^{(u,o,k)}_{n(u)})
  \Biggr].
  \label{eq:complete-ll}
\end{align}

\noindent\textbf{E-step:}
Given the current estimate
\(\phi^{(\tau)}=(\bm\pi,\{M_r\}_r,\bm{\theta})\) after \(\tau\)
iterations, define the forward vectors
\(\alpha_t(i)=\Prob(z_t=i\mid\phi^{(\tau)})\) and the backward vectors
\(\beta_t^{(u,o)}(i)=\Prob(o\mid z_t=i,u,\phi^{(\tau)})\).  These
vectors can be computed recursively according to
\begin{align}
  \alpha_0 &=\bm\pi,
  &
  \alpha_{t+1} &=Q\,\alpha_t,
  &&t=0,\ldots,n(u)-1.
  \label{eq:forward}\\
  \beta_{n(u)}^{(u,o)} &=M_{m(u)}^{\,T} e_o,
  &
  \beta_t^{(u,o)} &=Q^T\beta_{t+1}^{(u,o)},
  &&t=n(u)-1,\ldots,0.
  \label{eq:backward}
\end{align}
Eq.~\eqref{eq:forward} propagates
\(\alpha\) upward in \(t\), from the preparation at \(t=0\) toward the
measurement. Computing the sequence \(\alpha_0,\alpha_1,\ldots\) is
referred to as the forward pass. Because there is a single preparation, this
sequence is shared by every configuration and depends on \(u\) only
through the length \(n(u)\).  Eq.~\eqref{eq:backward}
propagates \(\beta\) backwards in \(t\), from the measurement at
\(t=n(u)\) back toward the preparation. Computing \(\beta\) is referred to as the
backward pass~\cite{baum1970maximization,ephraim2002hidden}.
Reading off the final forward vector gives the outcome distribution
\begin{align}
  \bm{p}_u =
  M_{m(u)}\alpha_{n(u)},
  \label{eq:pu-from-alpha}
\end{align}
which satisfies \(p_u(o)=\Prob(o\mid u,\phi^{(\tau)})\) and
depends on only the forward pass. 

The key ingredient in computing the auxiliary objective function for
the M-step is the probability distribution \(\mu\) of the hidden
random variables conditional on the observation given that the
parameters are \(\phi^{(\tau)}\).  The objective function is defined
as the expectation with respect to \(\mu\) of the complete-data
log-likelihood of Eq.~\eqref{eq:complete-ll}.  Each summand of the
complete-data log-likelihood depends on only one state or two
consecutive states in the latent state sequences, so it suffices to
compute marginal probabilities of these states.  For each
configuration \(u\) and outcome \(o\), the needed probabilities are
the single-state marginal probabilities \(\gamma_t^{(u,o)}\) which are
vectors over \(i\in B\), and the consecutive-pair marginal
probabilities \(\xi_t^{(u,o)}\), which are matrices over
\((i,j)\in B^2\). These are given by
\begin{align}
  \gamma_t^{(u,o)}(i)
  &= \Prob(z_t=i\mid o,u,\phi^{(\tau)})
   = \frac{\alpha_t(i)\,\beta_t^{(u,o)}(i)}{p_u(o)},
  \\
  \xi_t^{(u,o)}(i,j)
  &= \Prob(z_{t+1}=i,\,z_t=j\mid o,u,\phi^{(\tau)})
   = \frac{\alpha_t(j)\,Q_{ij}\,\beta_{t+1}^{(u,o)}(i)}{p_u(o)}.
  \label{eq:gamma-xi}
\end{align}

When computing the expectation of the complete-data log-likelihood of
Eq.~\eqref{eq:complete-ll}, each summand is multiplied by the appropriate
marginal conditional probability.  The sum of these probabilities over
the trials is the expected initial-state occupancy, the expected
number of \(j\to i\) transitions, and the expected outcome
numbers. These expectations are given by
\begin{subequations}\label{eq:suff-defn}
\begin{align}
  \Gamma^{\pi}_{i}
  &= \sum_{u\in\cU}\sum_{o=1}^O y_{u,o}\,\gamma_0^{(u,o)}(i)
     = \sum_{u\in\cU}\sum_{o=1}^O y_{u,o}\,\frac{\alpha_0(i)\,\beta_0^{(u,o)}(i)}{p_u(o)},
  \label{eq:Gpi-defn}\\
  N_{ij}
  &= \sum_{u\in\cU}\sum_{o=1}^O\sum_{t=0}^{n(u)-1}
    y_{u,o}\,\xi_t^{(u,o)}(i,j) =
    \sum_{u\in\cU}\sum_{t=0}^{n(u)-1}\sum_{o=1}^O
    y_{u,o}\,\frac{\alpha_t(j)\,Q_{ij}\,\beta_{t+1}^{(u,o)}(i)}{p_u(o)},
  \label{eq:N-defn}\\
  \Gamma^M_{r,o,i}
  &= \sum_{u:\,m(u)=r} y_{u,o}\,\gamma_{n(u)}^{(u,o)}(i)
    = \sum_{u:\,m(u)=r} y_{u,o}\,\frac{\alpha_{n(u)}(i)\,\beta_{n(u)}^{(u,o)}(i)}{p_u(o)}.
  \label{eq:GM-defn}
\end{align}
\end{subequations}
\(N_{ij}\) and \(\Gamma^{\pi}_{i}\) depend on \(\beta_{t}^{(u,o)}\) only through the
combined backwards vectors
\begin{align}
  b_t^{(u)}
  &\defeq
  \sum_{o=1}^O \frac{y_{u,o}}{p_u(o)}\,\beta_t^{(u,o)}.
  \label{eq:b-defn}
\end{align}
The \(b_{t}^{(u)}\) can be computed recursively the same way as the \(\beta_{t}^{(u,o)}\):
\begin{align}
  b_{n(u)}^{(u)} &= M_{m(u)}^{\,T}\!\left(\frac{\bm y_u}{\bm p_u}\right),
                   \qquad
                   b_t^{(u)} = Q^T b_{t+1}^{(u)}\quad(t=n(u)-1,\ldots,0),
                   \label{eq:b-recursion}
\end{align}
where \(\bm y_u/\bm p_u\) denotes elementwise division.  The E-step
depends on the \(\beta_{t}^{(u,o)}\) only through \(b_{t}^{(u)}\),
which are best computed directly with Eq.~\eqref{eq:b-recursion}.
Substituting the combined backwards vectors and the definitions of
\(\alpha_{0}(i)\) and \(\beta^{(u,o)}_{n(u)}\) gives
\begin{subequations}\label{eq:suff-stats}
  \begin{align}
    \Gamma^{\pi}_{i}
    &=\pi(i)\sum_{u\in\cU}b_0^{(u)}(i),
      \label{eq:Gpi-counts}\\
    N_{ij}
    &= Q_{ij}\sum_{u\in\cU}\sum_{t=0}^{n(u)-1}
      \alpha_t(j)\,b_{t+1}^{(u)}(i),
      \label{eq:N-counts}\\
    \Gamma^M_{r,o,i}
    &=\sum_{u:\,m(u)=r}\frac{y_{u,o}}{p_u(o)}\,
      \alpha_{n(u)}(i)\,
      (M_{m(u)}^{\,T}e_o)_i.
      \label{eq:GM-counts}
  \end{align}
\end{subequations}

The log-likelihood of the data with respect to the current estimate
is given by \(\ell(\phi^{(\tau)})=\sum_{u,o}y_{u,o}\log p_u(o)\),
which requires only the forward vectors to compute. Tests for
convergence of the EM algorithm require keeping track of
\(\ell(\phi^{(\tau)})\).

The auxiliary objective function of the candidate next estimate
\(\phi=(\bm\pi',\{M_r'\}_r,\bm{\theta}')\) is obtained by taking the
expectation of Eq.~\eqref{eq:complete-ll} over the latent state
sequences with respect to full conditional probability distribution
\(\prod_{u,o,k}\Prob(\bm z^{(u,o,k)}\mid u,o,\phi^{(\tau)})\). After
substituting the quantities obtained above, it is given by
\begin{align}
  \mathcal Q(\phi\mid\phi^{(\tau)})
  &=
  \underbrace{\sum_{i}\Gamma^{\pi}_{i}\log\pi'(i)}_{\defeq\,\mathcal L_\pi(\bm\pi')}
  +
  \underbrace{\sum_{r,o,i}\Gamma^M_{r,o,i}\log(M_r')_{o,i}}_{\defeq\,\mathcal L_M(\bm M')}
  +
  \underbrace{\sum_{i,j}N_{ij}\log Q_{ij}(\bm{\theta'})}_{\defeq\,\mathcal L_Q(\bm\theta')}.
  \label{eq:Q-objective}
\end{align}
We refer to summands on the right-hand side as the \(\pi\)-block, the
\(Q\)-block, and the \(M\)-block.

\noindent\textbf{M-step:}
Because the three blocks of Eq.~\eqref{eq:Q-objective} depend on
disjoint parameters of the candidate next estimate, they can be
maximized independently. The \(\pi\)- and \(M\)-blocks each reduce to
a box-constrained simplex problem, which we solve first. The
\(Q\)-block is a specialization of the update operation in the
Baum-Welch algorithm.

The box-constrained simplex problem is defined as follows:
Let \(a\in\rls^n_{\ge0}\), and let \(0\leq\ell_i\leq u_i\leq 1\) be bounds that satisfy
the strict inequalities \(\sum_i\ell_i<1<\sum_i u_i\).
The problem is to solve the optimization program
\begin{align}
  \max_{p\in\rls^n}\;\sum_i a_i\log p_i
  \quad\text{s.t.}\quad
  \mathbf 1^T p=1,\quad \ell_i\leq p_i\leq u_i,
  \label{eq:wf-problem}
\end{align}
which is a special case of the general nonlinear resource allocation
problem.  For example, see~\cite{patriksson2008survey}. We solved the
problem using the efficient, \(O(n\log(n))\) breakpoint search
algorithm for the continuous nonlinear knapsack
problem~\cite{bitran1981disaggregation,luss1975allocation}.  The
strict inequalities ensure Slater's condition for existence of a
feasible point \(p\) in the relative interior of the set satisfying
the constraints of the problem.

Our constraints on \(\bm{\pi}'\) and \(M'\) satisfy the definition of
the box-constrained simplex problem, so the optimization program of
Eq.~\eqref{eq:wf-problem} can be directly applied to maximize the
\(\pi\)-block over \(\bm{\pi}'\) and the \(M\)-block over \(M'\). For
the \(M\)-block, we apply the program separately to each column.  Here
are more details. To map the \(\pi\)-block optimization problem to
Eq.~\eqref{eq:wf-problem}, write \(a_i\defeq\Gamma^\pi_{i}\) with
\(i\in\{1,\ldots,|B|\}\) and \(p_i\defeq\pi(i)\). The lower bounds are
the \(\ell^\pi_{i}\in[0,1)\) described at the beginning of this
appendix, with \(\ell^{\pi}_{0}=0.8\) and \(\ell^{\pi}_{i}=0\) for
\(i\ne 0\).  We can set the upper bounds to \(u^{\pi}_{0}=1\). Column
stochasticity is enforced by the identity \(\bm{1}^{T}p=1\) in
Eq.~\eqref{eq:wf-problem}.

For the \(M\)-block, \(\mathcal L_M\) is a sum of terms each depending
only on a single column \((r,i)\) of the measurement matrices. Besides
entrywise bounds and column stochasticity constraints on \(M_r\), we
also have the identities of Eqs.~\eqref{eq:mr_twirl}
and~\eqref{eq:mr_twirl_1}. The \(M\)-block objective function
separates into independently maximizable objectives.  With the
parameterization given in the latter equations, the objective function
for \(M_{2}\) is
\[
  \underbrace{A \log a + B \log b + E \log e}_{a+b+e=1}
  \;+\;
  \underbrace{K \log c + F \log f}_{2c+f=1}
  \;+\;
  \underbrace{L \log d + J \log g}_{2d+g=1},
\]
\[
  \begin{aligned}
  A &= \Gamma_{00}+\Gamma_{11}, & B &= \Gamma_{01}+\Gamma_{10}, & E &= \Gamma_{20}+\Gamma_{21},\\
  K &= \Gamma_{02}+\Gamma_{12}, & F &= \Gamma_{22}, & & \\
  L &= \Gamma_{03}+\Gamma_{13}, & J &= \Gamma_{23}. & &
  \end{aligned}
\]
And for \(M_{1}\) the objective function is given by
\[
  \underbrace{A' \log a' + B' \log b' + G' \log g'}_{a'+b'+g'=1}
  \;+\;
  \underbrace{C' \log c' + E' \log e' + H' \log h'}_{c'+e'+h'=1}
  \;+\;
  \underbrace{D' \log d' + F' \log f' + J' \log j'}_{d'+f'+j'=1},
\]
\[
  \begin{aligned}
  A' &= \Gamma_{00}+\Gamma_{01}, & B' &= \Gamma_{10}+\Gamma_{11}, & G' &= \Gamma_{20}+\Gamma_{21},\\
  C' &= \Gamma_{02}, & E' &= \Gamma_{12}, & H' &= \Gamma_{22},\\
  D' &= \Gamma_{03}, & F' &= \Gamma_{13}, & J' &= \Gamma_{23}.
  \end{aligned}
\]
Each of the underbraced sums, along with the constraint equations underneath them, corresponds to an objective function that can be independently maximized as a non-linear resource allocation problem.

It remains to show how to optimize the \(Q\)-block.  Substituting
Eq.~\eqref{eq:Qreparam} into \(\mathcal L_Q\)  gives
\begin{align}
  \mathcal L_Q(\bm\theta)
  &=
  \opn{const} + A\log\me + B\log\teps + D_g\log l_g + D_e\log l_e\nonumber\\
  &+ C_g\log s_g + C_e\log s_e
  + E_g\log(1-s_g) + E_e\log(1-s_e),
  \label{eq:LQ}
\end{align}
where the coefficients are given in 
terms of the \(N_{ij}\) of Eq.~\eqref{eq:N-defn} as
\begin{align}
  A &= N_{00}+N_{11}, &
  B &= N_{01}+N_{10}, \\
  C_g &= N_{0g}+N_{1g}, &
  C_e &= N_{0e}+N_{1e}, \\
  D_g &= N_{g0}+N_{g1}, &
  D_e &= N_{e0}+N_{e1}, \\
  E_g &= N_{gg}, &
  E_e &= N_{ee}.
\end{align}
The leading constant ``\(\opn{const}\)'' in Eq.~\eqref{eq:LQ} arises from
the coefficients of \(1/2\) multiplying the seepage parameters and
shifts the objective without moving its maximizer. We drop this
constant below.

The constraints on \(\bm\theta\) are such that we can optimize
each of the following contributions to \(\mathcal L_{Q}(\bm{\theta})\) separately:
\begin{align}
  \cL_{Q,m}&=A\log\me + B\log\teps + D_g\log l_g + D_e\log l_e,\notag\\
  \cL_{Q,g}&=C_g\log s_g + E_g\log(1-s_g),\notag\\
  \cL_{Q,e}&=C_e\log s_e + E_e\log(1-s_e)
\end{align}
Optimizing \(\cL_{Q,g}\) and \(\cL_{Q,e}\) is accomplished by the Bernoulli
maximum likelihood estimators
\begin{align}
  s_g^\star = \frac{C_g}{C_g+E_g},
  \qquad
  s_e^\star = \frac{C_e}{C_e+E_e}.
  \label{eq:sstar}
\end{align}
Both are in \((0,1)\) provided the corresponding outcome counts are positive.
The  problem of optimizing \(\cL_{Q,m}\) is
\begin{align}
  \max_{\teps,l_g,l_e,\me\geq 0}\;
  A\log\me+B\log\teps+D_g\log l_g+D_e\log l_e
  \quad
  \text{s.t.}
  \quad
  \me+\teps+l_g+l_e=1
  \label{eq:Q-comp-subproblem}
\end{align}
Let \(N_{\mathrm c}\defeq A+B+D_g+D_e\).
We obtain the following closed-form solution
\begin{align}
  \me^\star = \frac{A}{N_{\mathrm c}},\quad
  \teps^\star = \frac{B}{N_{\mathrm c}},\quad
  l_g^\star = \frac{D_g}{N_{\mathrm c}},\quad
  l_e^\star = \frac{D_e}{N_{\mathrm c}}.
  \label{eq:Qstar}
\end{align}
The solutions in Eq.~\eqref{eq:sstar} and in Eq.~\eqref{eq:Qstar}
are the standard Baum-Welch update specialized to the parametrization
in Eq.~\eqref{eq:Qreparam}.

The M-step computes the next estimate of the parameters
\(\phi^{(\tau+1)}\) as the optimizer of
\(\mathcal Q(\phi\mid\phi^{(\tau)})\). The theory of EM algorithms for
hidden Markov models guarantees that this converges to a critical
point of the likelihood function for the observed data. See
Ref.~\cite{baum1970maximization}.

Convergence of the EM algorithm to a critical point can be slow.  Near
a fixed point \(\phi^\star\) of the EM map \(T:\cF\rightarrow\cF\)
defined by one E-step followed by the M-step of App.~\ref{app:em}, the
EM iteration converges linearly with an asymptotic rate that can be
very small. We accelerate convergence with a constrained SQUAREM-3
scheme~\cite{varadhan2008squarem,duvaradhan2020squarem}, modified to
keep every iterate feasible and ensure that the log-likelihood is
non-decreasing. The modification maintains convergence to a critical
point. The SQUAREM algorithm uses three iterates
\(\phi_{0}^{(\tau)}, \phi_{1}=T(\phi_{0}), \phi_{2}=T(\phi_{1})\) of
\(T\) and quadratic extrapolation to determine an improved estimate
\(\tilde\phi\) as a linear combination of \(\phi_{0}\), \(\phi_{1}\)
and \(\phi_{2}\). The estimate \(\tilde\phi\) satisfies all equalities
satisfied by all \(\phi_{i}\), such as probability normalization and the
measurement matrix identities of Eq.~\eqref{eq:mr_twirl}
and~\eqref{eq:mr_twirl_1}.  The next iterate is then
\(T(\tilde\phi)\), which requires that \(\tilde\phi\) is in the domain
of \(T\). This requirement need not be satisfied, that is, the
estimate \(\tilde\phi\) need not be feasible. To avoid this problem,
we take advantage of a domain extension of \(T\).  The constrained EM
map \(T:\mathcal F\to\mathcal F\) satisfies that for every
\(\phi\in\mathcal F\),
\begin{itemize}[nosep]
\item \(T(\phi)\in\mathcal F\),
\item \(\ell(T(\phi))\geq\ell(\phi)\),
\item \(T(\phi)=\phi\) only if \(\phi\) is a critical point of the
  constrained maximum likelihood problem of Eq.~\eqref{eq:constrained-mle}.
\end{itemize}
For convenience, we here treat \(\phi\) as a vector whose entries
include the entries of \(\bm{\pi}\) and the \(M_{r}\) as well as the
parameters \(\tilde\epsilon, l_{g},l_{e},s_{g},s_{e}\) that define
\(Q\).  The entries of this vector are non-negative and satisfy the
affine constraints determined by probability normalization, the
measurement matrix identities and
\(\tilde\epsilon+l_{g}+l_{e}\leq 1\), \(s_{g}\leq 1\) and
\(s_{e}\leq 1\).  We extend the domain of \(T\) to the set \(\cG\) of
values of \(\phi\) that satisfy these affine constraints.  \(\cG\)
does not respect the SPAM constraints in the definition of
\(\cF\). Inspection of the forward-backward definition of the
objective function defined in the E-step and optimized in the M-step
shows that the iteration is well defined even if \(\phi^{(\tau)}\) is in
\(\cG\), in which case the output of the M-step is in \(\cF\).

In SQUAREM, the estimate \(\tilde\phi\) is adjusted by a backtracking
method, basically until \(T(\tilde\phi)\) increases the log-likelihood
over that of the previous iterate \(\phi^{(\tau)}\). Our main
adjustment to SQUAREM as described in the reference is to force
backtracking if \(\tilde\phi\) is not in \(\cG\). If backtracking
fails after too many tries, the algorithm falls back to
\(\phi^{(\tau+1)}=\phi_{2}\).

\section{EM algorithm implementation notes}
\label{sec:notes}

We initialize the EM iterations with a random \(\phi^{(0)}\in\cF\),
where \(\phi^{(0)}\) is obtained from random bounded perturbations of
the reference point of App.~\ref{sec:coptimaldesign} ``projected''
back to \(\cF\). We used four random such starting points.
For  $Q$, we perturbed the reference point by adding Gaussian noise
of standard deviation \(0.005\) and ensured positivity by taking the maximum
of the resulting entries and \(10^{-6}\). We then normalized the columns
to ensure that their entries summed to one. From the resulting matrix \(Q'\),
we computed
\begin{align}
  m     &= \frac{1}{2}  (Q'_{00} + Q'_{11})\\
  \tilde{\epsilon} &= \frac{1}{2}(Q'_{01} + Q'_{10})\\
  l_g    &= \frac{1}{2}  (Q'_{20} + Q'_{21})\\
  l_e    &= \frac{1}{2} (Q'_{30} + Q'_{31})\\
  s_g &= (Q'_{02} + Q'_{12} + 2  (1 - Q'_{22})) / 3\\
  s_e &= (Q'_{03} + Q'_{13} + 2  (1 - Q'_{33})) / 3.
\end{align}
We need to
satisfy \(c_{1}= m+ \tilde{\epsilon}+ l_g+ l_e=1\), so we rescale
\((m,\tilde{\epsilon},l_{g},l_{e})\) by dividing by \(c_{1}\). We
obtain \(Q\) by reconstructing it from the final values of the
parameters \(m, \tilde{\epsilon},l_{g},l_{e},s_{g},s_{e}\).  For $\pi$
and $M$, to each entry of the reference point, we added a uniformly
distributed number on \((0, .005)\). To ensure strict positivity we
took the maximum of the resulting entries and \(10^{-12}\).  Since
\(\pi,M,Q\) so obtained need not satisfy all required constraints, we
projected them into the constrained set by applying a water filling
algorithm also used for the resource allocation problem with a
suitable objective function. If \(s_{e}<s_{g}\) we swapped \(s_{g}\) with \(s_{e}\) and \(l_{g}\)
with \(l_{e}\), thus ensuring that \(s_{e}\geq s_{g}\).  

\ifSubfilesClassLoaded{%
  \bibliographystyle{unsrt}
  \bibliography{leakagerb}
}{}

\end{document}


\ifSubfilesClassLoaded{
  \title{Determining parameter ranges by gauge search}
  \author{}
  \date{\today}
  \maketitle
}{
  \section{Determining parameter ranges by gauge search}
  \label{app:gauge search}
  \let\section\subsection
}

The model given in the main text is not identifiable. Distinct
parameter values can assign the same likelihood to every possible data
set, and two such values are called gauge-equivalent.  Given a
candidate parameter value such as the one returned by the EM algorithm
of App.~\ref{app:em}, we aim to calculate a set of feasible parameters
that contains all valid parameters that are gauge-equivalent to the
candidate. This bounding set defines constraints on the
gauge-equivalent parameters that contribute to the reported ranges of
inferred parameters.  Here we describe an algorithm that implements
the ideas outlined in Sect.~\ref{sec:gaugeopt} for computing such a
bounding set.

Let \(\pi_{\ML}\), \(M_{1,\ML}\), \(M_{2,\ML}\) and \(Q_{\ML}\) be
the maximum likelihood fit to the data, where the fit satisfies the
imposed constraints. We do not take advantage of the constraints on
the \(M_{r}\) in Eqs.~\eqref{eq:mr_twirl} and~\eqref{eq:mr_twirl_1}
implied by the randomized final permutation.  The fit determines the
gauge free matrices \(Y\) and \(A\) as defined in
Sect.~\ref{sec:gaugeopt}.  A gauge equivalent transition matrix \(Q\)
has the same eigenvalues as \(Y\) and \(Q_{\ML}\). The eigenvalues are
given by \(1,\lambda_{1},\lambda_{-}\geq \lambda_{+}\) and determined
in Sect.~\ref{sec:stochasticmodel}. For every gauge equivalent \(Q\),
its parameters \(\bm{f}=(\epsilon,l_{g},l_{e},s_{g},s_{e})\) are a
function of \(\theta=(l_{g},s_g)\), where \(\bm{f}(\theta)\) is
obtained by applying Eq.~\eqref{eq:lgsg}. Thus we can write
\(Q=Q(\theta)\). In addition to the variables introduced in
Eq.~\eqref{eq:lgsg} we introduce \(D=1-l_{g}-l_{e}\) for the probability
of remaining in the computational subspace and
\(\tilde\epsilon=\epsilon D\).  The diagonal and off-diagonal entries
of the computational block of \(Q\) are then given by
\(m=(1-\epsilon)D\), \(\tilde\epsilon\), which satisfy
\(m+\tilde\epsilon=D\), and \(m-\tilde\epsilon=\lambda_{1}\). We have
the identities
\begin{align}
  m(\theta)=\tfrac12\bigl(D(\theta)+\lambda_{1}\bigr),\qquad
  \tilde\epsilon(\theta)=\tfrac12\bigl(D(\theta)-\lambda_{1}\bigr),\qquad
  \epsilon(\theta)=\frac{\tilde\epsilon(\theta)}{D(\theta)}.
  \label{eq:m-teps}
\end{align}

It is shown in Sect.~\ref{sec:gaugeopt} that every
\(\pi, M_{1}, M_{2}, Q\) that are gauge equivalent to the maximum
likelihood fit are determined by a choice of \(\theta\) and a
\(4\times 4\) real invertible diagonal matrix \(\Sigma\) with
diagonal \(\bm{\sigma}=(\sigma_{1},\ldots,\sigma_{4})\).  Let \(M\)
be the combined measurement matrix defined in Sect.~\ref{sec:gaugeopt}. 
In terms of \(\bm{\sigma}\) and
\(\theta\) we have
\begin{align}
  M(\bm\sigma;\theta)=V\Sigma W(\theta)^{-1},
  \qquad
  \pi(\bm\sigma;\theta)=W(\theta)\Sigma^{-1}V^{-1}A.
  \label{eq:gauge-orbit}
\end{align}

The eigenvectors of \(Q\) are described in
Sect.~\ref{sec:stochasticmodel}. 
Explicit right eigenvectors \(w(\lambda)\) for eigenvalue \(\lambda\)
are  $w(\lambda_{1}) = (1, -1, 0, 0)^T$ and for
\(\lambda\in\{1,\lambda_{+},\lambda_{-}\}\) 
\begin{align}
  w(\lambda)=
  \begin{pmatrix}
    \tfrac12(\lambda-1+s_{g})(\lambda-1+s_{e})\\
    \tfrac12(\lambda-1+s_{g})(\lambda-1+s_{e})\\
    l_{g}(\lambda-1+s_{e})\\
    l_{e}(\lambda-1+s_{g})
  \end{pmatrix},
  \label{eq:wlambda}
\end{align}
where these forms can be obtained by directly solving for the
eigenvectors of \(R\) given in Eq.~\eqref{eq:Rinvariant}.
Specifically, the three-dimensional invariant subspace consists of
vectors of the form \((\tfrac{a}{2},\tfrac{a}{2},y,z)^{T}\).  The
restriction of \(Q\) acts on \((a,y,z)\) as
\begin{align}
  R=\begin{pmatrix}
    D & s_{g} & s_{e}\\
    l_{g} & 1-s_{g} & 0\\
    l_{e} & 0 & 1-s_{e}
  \end{pmatrix},
\end{align}
which has spectrum \(\opn{spec}(R)=\{1,\lambda_{+},\lambda_{-}\}\).
For \(\lambda\in\opn{spec}(R)\) the last two rows of
\((R-\lambda I)(a,y,z)^{T}=0\) give
\((\lambda-1+s_{g})\,y=l_{g}a\) and
\((\lambda-1+s_{e})\,z=l_{e}a\). The first row of the eigenvalue
identity is satisfied because \(\lambda\in\opn{spec}(R)\). Setting
\(a=(\lambda-1+s_{g})(\lambda-1+s_{e})\) gives
\(y=l_{g}(\lambda-1+s_{e})\) and \(z=l_{e}(\lambda-1+s_{g})\),
from which we obtain the vector of Eq.~\eqref{eq:wlambda}.

Not every value of \(\theta\) yields a valid \(\bm f(\theta)\). The parameters
\(l_{e},s_{e}\) in \(\bm{f}(\theta)\) must lie in \([0,1]\) and the values
\(m,\tilde\epsilon\) determined by \(\bm{f}(\theta)\) must be nonnegative. 
The resulting constraints are equivalent to requiring that 
\begin{align}
  r(\theta)=\bigl(l_{e},\,1-l_{e},\,s_{e},\,1-s_{e},\,m,\,\tilde\epsilon\bigr)
  \label{eq:r}
\end{align}
be entrywise nonnegative. Positivity of \(\tilde\epsilon\), along with $\lambda_1\ge 0$, also guarantees that \(l_{g}+l_{e}\leq 1\). 
We refer to the set of \(\theta\)
satisfying these constraints as the physical region, which is
given by
\begin{align}
  \cA^{\mathrm{phys}}
  =\bigl\{\theta=(l_{g},s_{g}):\;l_{g}> 0,\;s_{g}\geq 0,\; r(\theta)\geq 0\bigr\},
  \label{eq:Aphys}
\end{align}
where the last inequality is taken entrywise. The strict constraint on
\(l_{g}\) is to prevent a zero denominator when computing \(l_{e}\) as
a function of \(\theta\).

It suffices to confine the search for feasible pairs \(\theta\) to the
physical region \(\cA^{\mathrm{phys}}\).  The set of feasible pairs
need not be connected. The relabelling
\((l_{g},s_{g})\leftrightarrow(l_{e},s_{e})\) that exchanges ground
and excited states permutes the two leaked states but does not change
the spectrum \(\{1,\lambda_{1},\lambda_{+},\lambda_{-}\}\).  For large
seepage it can map the connected component of the feasible region that
contains the truth to a different connected component.

It is possible to obtain analytic expressions for the boundary of
\(\cA^{\mathrm{phys}}\). But to simplify the search and obtain an
initial grid of candidate feasible pairs \(\theta\) we determine
an enclosing band \(\cB\) with linear boundary segments that
contains \(\cA^{\mathrm{phys}}\). The enclosing band is defined in
the following proposition.

\begin{proposition}
  \label{prop:band}
  Every \(\theta\in\cA^{\mathrm{phys}}\) satisfies
  \(\delta_{-}\leq l_{g}+s_{g}\leq\delta_{+}\).
  Hence
  \begin{align}
    \cA^{\mathrm{phys}}\subseteq\cB
    :=\bigl\{(l_{g},s_{g}):\;l_{g}\geq 0,\;s_{g}\geq 0,\; \delta_{-}\leq l_{g}+s_{g}\leq\delta_{+}\bigr\}.
    \label{eq:band}
  \end{align}
\end{proposition}

\begin{proof}
  Let \(\theta\in\cA^{\mathrm{phys}}\), so \(l_{g}>0\) and, by
  Eq.~\eqref{eq:lgsg}, \(l_{e}=\bigl(c-(2t_{g}-b)^{2}\bigr)/(4l_{g})\geq 0\).
  Thus \((2t_{g}-b)^{2}\leq c\). Because \(b=\delta_{+}+\delta_{-}\) and
  \(\sqrt{c}=\delta_{+}-\delta_{-}\geq 0\) we get
  \(\lvert 2t_{g}-(\delta_{+}+\delta_{-})\rvert\leq\delta_{+}-\delta_{-}\),
  which implies that \(\delta_{-}\leq t_{g}=l_{g}+s_{g}\leq\delta_{+}\). Together with
  \(l_{g},s_{g}\geq 0\) these are the defining inequalities for \(\cB\).
\end{proof}

We cover the band \(\cB\) with an initial grid of candidate pairs
\(\theta\) and then apply a nested strategy for eliminating
infeasible pairs.  The grid is laid so that a grid step changes each
transition parameter by a controlled amount. Let $J(\theta)$ be the
\(5\times 2\) Jacobian of the parameters \(\bm{f}(\theta)\) with
respect to \(\theta\). The Jacobian is evaluated once, at the
maximum likelihood fit, which lies in \(\cA^{\mathrm{phys}}\).
Write the singular value decomposition of $J$ as
\(J=\sum_{i=1,2}\varsigma_{i}\,p_{i}q_{i}^{T}\), with singular
values \(\varsigma_{i}\), right singular vectors \(q_{i}\) and left
singular vectors \(p_{i}\). A step of length
\(\delta/\varsigma_{i}\) along \(q_{i}\) changes \(\bm f\) by
approximately \(\delta\,p_{i}\).  Spacing the grid by
\(\delta/\varsigma_{i}\) along \(q_{i}\) therefore gives to first order that
neighbouring grid points differ by at most \(\delta\) in each
transition parameter.

The singular vectors \(q_{1},q_{2}\) form an orthonormal basis of
the \(\theta\)-plane, so the grid is a rectangle in this frame,
which is in general rotated relative to the \((l_{g},s_{g})\)
axes. We size it to cover the bounding box \([0,\delta_{+}]^{2}\)
of \(\cB\).  We choose the center of the grid to be
\(\kappa=(\delta_{+}/2,\delta_{+}/2)\).  A point
\(\kappa+(x_{1},x_{2})\) in \(\cB\) has \(|x_{1}|,|x_{2}|\leq\delta_{+}/2\).
Define the half-extent
\begin{align}
  H_{i}=\tfrac{\delta_{+}}{2}\bigl(|q_{i,1}|+|q_{i,2}|\bigr).
\end{align}
Running the grid from \(-H_{i}\) to \(H_{i}\) about \(\kappa\)
along each \(q_{i}\) thus reaches every point of the bounding box,
hence of \(\cB\).  At spacing \(\delta/\varsigma_{i}\) this is
\(\lceil 2H_{i}\varsigma_{i}/\delta\rceil\) steps along
\(q_{i}\).  The grid construction is completed by
discarding candidate pairs \(\theta\) in the grid that do not lie
in \(\cA^{\mathrm{phys}}\). In this way we obtain the set of candidate pairs
\begin{align}
  \cS=\bigl\{\theta\ \text{in the grid}:l_{g}> 0,\; s_{g}\geq 0,\;r(\theta)\geq 0\}.
  \label{eq:S}
\end{align}

We screen the remaining candidate pairs \(\theta\) in
\(\cA^{\mathrm{phys}}\) for feasibility by checking necessary
conditions for there to exist \(\bm{\sigma}\) with no zero entries
such that \(M(\bm{\sigma};\theta)\) and
\(\pi(\bm{\sigma};\theta)\) satisfy the constraints on \(M\) and
\(\pi\). 

We first check that the matrix of right eigenvectors
\(W(\theta)\) has full rank and is well conditioned, and eliminate
\(\theta\) if not. This ensures that the expression for
\(M(\bm{\sigma};\theta)\) can be evaluated for all
\(\bm{\sigma}\).  In view of the form of the unnormalized
eigenvectors given in Eq.~\eqref{eq:wlambda}, we also eliminate
\(\theta\) if we have both \(\lambda-1+s_{g}=0\) and
\(\lambda-1+s_{e}=0\) for
\(\lambda\in\{1,\lambda_{+},\lambda_{-}\}\). 

The remaining checks are based on constraints in each of the
\(16\) orthants for \(\bm{\sigma}\) derived from the measurement
and preparation constraints.  The measurement quality constraints are
linear constraints on the entries of \(M(\bm{\sigma};\theta)\),
which in turn are linear in \(\bm{\sigma}\) according to
\begin{align}
  M_{kl}=\sum_{i}V_{ki}\,(W(\theta)^{-1})_{il}\,\sigma_{i}.
  \label{eq:Mkl_theta}
\end{align}
The measurement quality constraints can therefore be written in
the form
\begin{align}
  \ell^{M}\leq A_{M}(\theta)\,\bm\sigma\leq u^{M},
  \label{eq:Mbox}
\end{align}
where \(A_{M}(\theta)\) is the \(16\times 4\) matrix that collects
the coefficients of \(\bm{\sigma}\) in Eq.~\eqref{eq:Mkl_theta},
and the fixed bounds \(\ell^{M},u^{M}\) are the SPAM quality
bounds collected into \(16\)-dimensional vectors.  The
inequalities in Eq.~\eqref{eq:Mbox} hold entrywise.  Additional
constraints on the measurement matrix are due to
column-stochasticity of the full measurement matrices, which implies
the bounds
\begin{align}
  M_{0l}+M_{1l}\leq 1,\qquad M_{2l}+M_{3l}\leq 1\qquad\text{for each column }l.
  \label{eq:Mcolsum}
\end{align}

The constraints on the preparation include quality constraints of
the form \(\ell^{\pi}\leq\pi(\bm{\sigma};\theta)\leq u^{\pi}\) and
the normalization \(\mathbf 1^{T}\pi=1\). Because
\(\pi(\bm{\sigma};\theta)=W(\theta)\Sigma^{-1}V^{-1}A\), these
constraints are non-linear.  From this expression for
\(\pi(\bm{\sigma};\theta)\) we obtain the identity
\(\sigma_{i}=(V^{-1}A)_{i}/(W(\theta)^{-1}\pi)_{i}\).  Let
\(d_{i}=(W(\theta)^{-1}\pi)_{i}\). Since \(d_{i}\) is a linear
function of \(\pi\), we can determine the maximum and minimum
values \(d_{i}^{+}\) and \(d_{i}^{-}\) of \(d_{i}\) over the
polytope of \(\pi\) satisfying the quality constraints and
stochasticity.  By monotonicity of the function
\(d_{i}\mapsto (V^{-1}A)_{i}/d_{i}\) on each side of \(d_{i}=0\),
if \(d_{i}^{-}\) and \(d_{i}^{+}\) have the same sign, then the
values of \(\sigma_{i}\) are constrained to lie in the interval
\((V^{-1}A)_{i}[1/d_{i}^{+},1/d_{i}^{-}]\), where this interval is the
set obtained by multiplying each member of the given interval by
\((V^{-1}A)_{i}\). Here, \((V^{-1}A)_{i}\) may be negative or
zero. If \(d_{i}^{-}\) and \(d_{i}^{+}\) do not have the same
sign, then the implied constraint splits into two rays depending
on the sign of \(\sigma_{i}\). We apply these constraints when
considering each orthant in the last check below.

Based on the above reasoning, we next eliminate candidates
\(\theta\) for which it is not possible to satisfy the linear
constraints on \(\bm{\sigma}\) determined by the constraints on
\(M\) and the interval constraints on the \(\sigma_{i}\) obtained
from \(d_{i}^{-}\) and \(d_{i}^{+}\) when the latter have the same
sign. This is a linear program feasibility test. In our testing,
this eliminates a large fraction of the remaining candidates.

The last check requires iterating over all \(16\) orthants for
\(\bm{\sigma}\).  Fix a sign pattern \(\bm s\in\{\pm1\}^{4}\) and
consider the orthant on which \(\opn{sgn}(\sigma_{i})=s_{i}\) for
each \(i\).  Let \(s'_{i}=\opn{sgn}\bigl((V^{-1}A)_{i}\bigr)\),
which depends only on the data.  In this orthant, we have
\(\opn{sgn}(\qty(W^{-1}\pi)_{i})=s'_{i}s_{i}\).
In combination with the preparation constraints, the preparations consistent
with this orthant for \(\bm{\sigma}\) are in the polytope
\begin{align}
  \cP(\theta;\bm s)=\left\{\pi:\;
  \ell^{\pi}\leq\pi\leq u^{\pi},\;
  \mathbf 1^{T}\pi=1,\;
  (s'_{j}s_{j})\,(W(\theta)^{-1}\pi)_{j}\geq 0\ \ \forall j\right\}.
  \label{eq:Ppoly}
\end{align}
On \(\cP(\theta;\bm s)\), \(d_{i}\) has fixed sign and we can refine
the lower and upper bound calculation for \(d_{i}\) by solving two linear
programs to determine the lower and upper bound of \(d_{i}\) on \(\cP(\theta;\bm s)\).
It is possible for one of the implied bounds on $\sigma$ to be signed infinity. We obtain
\(d_{i}\in\qty[\min_{\pi\in\cP(\theta;\bm s)}d_{i},\,\max_{\pi\in\cP(\theta;\bm s)}d_{i}]\).
By monotonicity on the orthant of the relationship between \(d_{i}\) and \(\sigma_{i}\),
we obtain a vectorized box constraint
\begin{align}
  \ell^{\sigma}(\theta;\bm s)\leq\bm\sigma\leq u^{\sigma}(\theta;\bm s)
  \label{eq:sigmabox}
\end{align}
on \(\bm{\sigma}\).
The last check is completed by checking whether the linear constraints on
\(M\) and the box constraints of Eq.~\eqref{eq:sigmabox} are feasible.
If they are not, we discard \(\theta\).

The set \(L\) of remaining feasible \(\theta\) in the grid
determines the gauge-compatible parameter intervals that we
report for a given experiment as described in
Sect.~\ref{sec:gaugeopt}.

\ifSubfilesClassLoaded{%
  \bibliographystyle{unsrt}
  \bibliography{leakagerb}
}{}

\end{document}


\end{document}